\documentclass[letterpaper,twocolumn,10pt]{article}
\usepackage{usenix}
\usepackage{amsfonts,amssymb,amsmath,amsthm, booktabs}
\usepackage{algorithmicx,xcolor}
\usepackage{colortbl}
\usepackage{algorithm}
\usepackage{algpseudocode}
\usepackage{mdframed}
\usepackage{comment}
\usepackage{float}
\usepackage{amsthm}
\usepackage{amsmath,amssymb}
\usepackage{enumitem}
\setlistdepth{10}
\usepackage{booktabs}
\usepackage{graphicx} 
\usepackage{upgreek}
\usepackage{tikz}
\usetikzlibrary{
    arrows.meta,
    positioning
}
\usepackage{booktabs}
\usepackage{tabularx}
\usepackage{array}
\usepackage{multirow}
\usepackage{graphicx}
\usepackage{pgfplots}
\usepgfplotslibrary{colormaps}
\pgfplotsset{compat=1.18}

\usepackage{tikz}
\usepackage{amsmath}

\usepackage{enumitem} 
\newcommand\setItemnumber[1]{\setcounter{enum\romannumeral\@enumdepth}{\numexpr#1-1\relax}}

\usepackage{enumitem, amsmath, amssymb}
\usepackage{stmaryrd}
\usepackage{makecell}
\usepackage{tikz}
\usepackage{xurl}
\usepackage{hyperref}
\usepackage{todonotes}

\usepackage{varwidth}
\usepackage[operators,sets]{cryptocode}
\usepackage[breakable]{tcolorbox}
\usepackage{cuted}
\usepackage{subcaption}
\newcommand{\sysname}{\texorpdfstring{\textsc{Pim\ensuremath{\epsilon}nto}}{PIMεNTO}}
\newcommand{\Del}[1]{\ensuremath{\textsc{Del}_{#1}}}
\newcommand{\Ins}[1]{\ensuremath{\textsc{Ins}_{#1}}}
\newcommand{\Chg}[1]{\ensuremath{\textsc{Chg}_{#1}}}

\newcommand{\boundary}[1]{\ensuremath{\mathsf{B}_{#1}}}
\newcommand{\multiplicity}[2]{\ensuremath{\mathrm{m}_{#1}(#2)}}
\newcommand{\witness}[1]{\ensuremath{T_{E_{#1}}}}
\newcommand{\mufree}[1]{\ensuremath{\mu^{\mathrm{free}}_{i,#1}}}
\newcommand{\mufreemin}{\ensuremath{\mu^{\mathrm{free}}_{i,\min}}}
\newcommand{\Munocci}{\ensuremath{M^{\mathrm{unocc}}_i}}
\newcommand{\Unique}{\ensuremath{\mathsf{Unique}}}
\newcommand{\select}[2]{\ensuremath{\sigma_{#1}\!\left(#2\right)}}
\newcommand{\project}[2]{\ensuremath{\pi_{#1}\!\left(#2\right)}}
\newcommand{\projectselect}[3]{%
  \ensuremath{\pi_{#1}\!\left(\sigma_{#2}\!\left(#3\right)\right)}%
}

\newcommand{\GS}[1]{\ensuremath{\mathrm{GS}_{#1}}}
\newcommand{\LS}[1]{\ensuremath{\mathrm{LS}_{#1}}}
\newcommand{\RS}[1]{\ensuremath{\mathrm{RS}_{#1}}}

\newcommand{\NLQ}{\ensuremath{x}}                 
\newcommand{\PrivDB}{\ensuremath{D}}              
\newcommand{\SynDB}{\ensuremath{\widetilde D}}    
\newcommand{\SynMech}{\ensuremath{\mathcal{M}_{\mathrm{syn}}}} 
\newcommand{\CandSet}{\ensuremath{\mathcal{Q}(\NLQ,\Sch)}}      
\newcommand{\OptQ}{\ensuremath{q^\star}}          
\newcommand{\ResRS}[1]{\ensuremath{\RS{#1}(\SynDB)}} 
\newcommand{\PrivRels}{\ensuremath{\mathcal{R}_{\mathrm{priv}}}}
\newcommand{\PruneFrac}{\ensuremath{\alpha_{\mathrm{prune}}}}
\newcommand{\DistSet}[1]{\ensuremath{\mathfrak{S}_{#1}}}

\newcommand{\ActSetCount}{\textsc{Set\_Count}}
\newcommand{\ActAddRelation}{\textsc{Add\_Relation}}
\newcommand{\ActStop}{\textsc{Stop}}

\newcommand{\Incumbent}[1]{\ensuremath{\iota_{#1}}} 

\newcommand{\RewCorr}[1]{\ensuremath{R_{\mathrm{corr}}(#1)}}

\newcommand{\TSQL}{\texttt{T2SQL}}

\algnewcommand{\Input}{\State \textbf{Input:} }
\algnewcommand{\Output}{\State \textbf{Output:} }

\newenvironment{packeditemize}{
\begin{itemize}[
  leftmargin=0.4em,
  labelwidth=0em,
  labelsep=0.2em
]
  \setlength{\itemsep}{0pt}
  \setlength{\parskip}{0pt}
  \setlength{\parsep}{0pt}
}{\end{itemize}}

\newenvironment{packeditemize2}{
\begin{itemize}[
  leftmargin=1.8em,
  labelwidth=1.1em,
  labelsep=0.4em
]
  \setlength{\itemsep}{0pt}
  \setlength{\parskip}{0pt}
  \setlength{\parsep}{0pt}
}{\end{itemize}}

\newtheorem{definition}{Definition}
\newtheorem{corollary}{Corollary}
\newtheorem{lemma}{Lemma}

\newenvironment{mybox2}[1]{%
    \begin{tcolorbox}[title={#1}, 
    colback=white, breakable]%
    }{
    \end{tcolorbox}
}

\newcounter{algo}
\renewcommand{\thealgo}{\arabic{algo}}

\newcommand{\algotitle}[2]{\refstepcounter{algo}\label{#1}Alg.~\thealgo: #2}
\definecolor{cSel}{RGB}{212,237,218}
\definecolor{cDer}{RGB}{255,242,204}
\definecolor{cFix}{RGB}{204,229,255}
\definecolor{cDen}{RGB}{248,215,218}

\newcommand{\Sch}{\ensuremath{\mathcal S}}          
\newcommand{\Pol}{\ensuremath{\mathcal P}}          
\newcommand{\Vset}{\ensuremath{\mathcal V}}         

\newcommand{\Vsnd}[1]{\ensuremath{V_{#1}^{\mathrm{snd}}}}
\newcommand{\Vrec}[1]{\ensuremath{V_{#1}^{\mathrm{rec}}}}
\newcommand{\Vsub}[1]{\ensuremath{V_{#1}^{\mathrm{sub}}}}
\newcommand{\Vdat}[1]{\ensuremath{V_{#1}^{\mathrm{data}}}}
\newcommand{\Vth}[1]{\ensuremath{V_{#1}^{\theta}}}

\newcommand{\Vex}{\ensuremath{V_\rho^{\mathrm{exact}}}}

\newcommand{\Req}{\ensuremath{\rho}}
\newcommand{\ReqQuery}{\ensuremath{q_\rho}}

\newcommand{\Pos}{\ensuremath{\mathsf{Pos}}}
\newcommand{\Viol}{\ensuremath{\mathsf{Viol}}}

\newcommand{\dataset}[1]{\textsc{Context-D#1}}
\newcommand{\dpdataset}{\textsc{DP-D1}}

\newcommand{\public}{\ensuremath{\mathcal{A}_c(D)}}

\newcommand{\PrivRelsC}{\ensuremath{\PrivRels^{c}}}        
\newcommand{\CLS}[1]{\ensuremath{\mathrm{CLS}_{#1}}}       
\newcommand{\CLShat}{\ensuremath{\widehat{\mathrm{CLS}}}}  
\newcommand{\CRS}[1]{\ensuremath{\mathrm{CRS}_{#1}}}       

\newtheorem{theorem}{Theorem}
\hypersetup{
  pdftitle={\sysname: A Privacy Framework for Querying Text},
  pdfauthor={Mushtari Sadia, Ang Chen, Amrita Roy Chowdhury}
}
\begin{document}

\date{}


\title{\Large \bf \sysname: A Privacy Framework for Querying Text}

\author{
  {\normalfont Mushtari Sadia\textsuperscript{*} \qquad Ang Chen\textsuperscript{*} \qquad Amrita Roy Chowdhury\textsuperscript{*}}\\
  {\normalfont\textsuperscript{*}University of Michigan, Ann Arbor}
}

\maketitle

\begin{abstract}
Currently, there are two state-of-the-art, complementary privacy guarantees: contextual integrity (CI) for what may flow, and differential privacy (DP) for what may be inferred. Yet neither maps cleanly onto natural language, leaving existing approaches unable to provide these guarantees for analytics over unstructured text. We address this gap with \sysname{}, a framework that takes three forms of natural language: text corpus, queries, and privacy policies; and grounds them into a relational database, creating a common substrate on which both guarantees can be enforced formally. With this design, we not only provide end to end privacy guarantees, but also improvement to utility through three key contributions: \emph{DP aware Text-to-SQL}, which searches for correct queries requiring the least DP noise; \emph{CI aware Text-to-SQL}, which compiles natural language policies into executable CI rules over the database; and a new privacy definition we call \emph{contextual differential privacy}, which redefines the traditional DP neighborhood under CI, and yields a tighter smooth sensitivity bound. Across new benchmarks, \sysname{} selects the best query in 75.3\% of cases (upto +45 points over baselines) and achieves zero leakage under correct policy grounding. To our knowledge, \sysname{} is the first framework to provide formal privacy guarantees for natural language analytics under CI, DP, and their composition.

\end{abstract}

\section{Introduction}\label{sec:intro} 

Organizations increasingly sit on vast amounts of rich information that is trapped in unstructured text -- as much as 80\% of enterprise data~\cite{unstructured-data}. Hospitals store physician notes, admission reports, medication records, and discharge summaries; law firms maintain case files; companies retain internal communications; public agencies collect similarly sensitive textual data. Across these domains, many important analyses reduce naturally to statistical queries over entities and relationships described in text. 
Consider a hospital researcher asking: ``How many patients reported worsening symptoms after starting a new HIV medication?'' Answering this query requires reasoning and linking evidence scattered across a corpus of documents; the answer cannot simply be read from a single record. Such analyses could be enormously valuable for research, operations, and policy. 
Yet much of this value remains locked away; the underlying text is highly sensitive, its use is constrained by stringent legal and institutional requirements, such as HIPAA~\cite{wikipedia_hipaa} and GDPR~\cite{wikipedia_gdpr}, and meaningful analysis is only viable with strong privacy guarantees.

Formalizing privacy directly over free-form text, however, is fundamentally difficult. Sensitive information has no fixed representation: it may span sentences, appear through paraphrases, be implied indirectly, or emerge only by combining innocuous facts. There is no obvious ``unit'' that captures what must be protected. More fundamentally, privacy has two dimensions. First, it is contextual: whether information may be revealed depends on who is asking, what is queried, and under what relationship or policy. Second, even an appropriate aggregate query can reveal information about individuals~\cite{dwork2006differential}. No single privacy notion naturally captures both risks. Returning to our hospital example, the query about worsening symptoms may be appropriate for a medical researcher conducting an approved HIV study, but not for a budget analyst unauthorized to access patients' HIV information. Even the authorized researcher should not receive an exact answer if it could reveal whether a patient is in the queried population.

These risks correspond to two foundational privacy frameworks studied extensively, but independently. Contextual integrity (CI)~\cite{barth2006privacy} is normative and governs the appropriateness of an information flow based on the actors, information being queried, and relevant relationships and conditions. In our example, CI may permit the query for a researcher associated with an approved HIV study while denying the same query from a budget analyst. Differential privacy (DP)~\cite{dwork2006differential}, in contrast, limits information revealed by adding noise calibrated to each individual's influence on the query answer. Thus, the two guarantees answer different questions about the \emph{same query}: CI determines \textit{what} information may flow; DP determines \textit{how much} that flow reveals.

Unfortunately, existing approaches to enforcing either CI or DP are fundamentally \textit{incompatible} with natural language text. For CI, the formal framework is precise~\cite{barth2006privacy}: policies can be expressed in first-order logic over actors, attributes, relationships, and transmission principles. The difficulty is enforcing them over raw text, where these objects are implicit and must be inferred from context. Consequently, prior work has largely relied on LLMs or heuristic contextual classifiers to judge whether information flows are appropriate~\cite{mireshghallah2024can,ghalebikesabi2025privacy,tsai2025contextual,yi2025privacy,10.1145/3173574.3173842,lan2025contextual}, or on rewriting, minimization, and structured communication layers~\cite{bagdasarian2024airgapagent,siyan2024papillon,ngong2025protecting,abdelnabi1822firewalls,ghalebikesabi2025privacy}. These approaches are empirical and provide no end-to-end formal guarantee that the resulting disclosure satisfies the CI policy.

Differential privacy faces a similar barrier, as it requires a well-defined notion of an individual's contribution. These concepts are natural for structured data such as relational databases, but have no direct analogue in raw text, where information about an individual may be scattered across tokens, sentences, and documents. Existing approaches therefore resort to a weaker relaxation of DP, treating a token or embedding as the protected unit~\cite{meisenbacher-etal-2024-comparative,feyisetan2020privacy,xu-etal-2020-differentially,feyisetan-kasiviswanathan-2021-private}, which no longer corresponds to an individual. Even setting this mismatch aside, these mechanisms are impractical: DP in high-dimensional embedding spaces can require enough noise to destroy semantic utility~\cite{yue2021differentialprivacytextanalytics}, while local-DP sanitization yields coherent text only at privacy parameters too weak for meaningful protection~\cite{meisenbacher2025leveragingsemantictriplesprivate}. LLM-based rewriting is also heuristic and provides no formal privacy guarantee~\cite{klymenko-etal-2022-differential,hu-etal-2024-differentially,igamberdiev2023dp}.


\noindent \textbf{A new privacy substrate for text.}
Our key insight is to change the representation on which privacy is enforced. Rather than defining CI and DP directly over text, we introduce a structured representation, \textit{relational databases}, between natural language and privacy enforcement. This builds on a broader shift in academia~\cite{liang2026beyond,lin2026structure,deng-etal-2024-text,arora2025languagemodelsenablesimple} and industry~\cite{google_document_ai,microsoft_content_understanding,microsoft_document_intelligence,unstructured_extract,talonic,unstract_aws,snowflake_unstructured_data}, where unstructured text is mapped into a database through \emph{Text-to-DB}~\cite{squid}, and natural-language questions are compiled into executable SQL through \emph{Text-to-SQL}~\cite{zhu2024large,hong2025next,shi2025survey,mohammadjafari2024natural}. This paradigm offers the best of both worlds: users retain the flexibility of natural language, while decades of mature database machinery can be brought to bear. The need for structure has also been recognized at the policy level, as recent amendments to CCPA~\cite{reuters2022ccpa,congruity360ccpa} call for bringing ``structure to unstructured data'' to enable privacy analysis.

Our observation is that this same representation unlocks a qualitatively new benefit: \textit{formal privacy guarantees over text}. Relational databases provide the explicit objects that both CI and DP lack over raw language. For DP, the connection is immediate. Once a natural-language question is translated into SQL, its privacy cost becomes a mathematical property of an explicit query over an explicit database, allowing us to leverage mature DP theory for relational queries~\cite{johnson2018towards,wilson2019differentially,kotsogiannis2019privatesql,bater2018shrinkwrap}. For CI, we make a complementary observation. A CI norm specifies predicates over contextual parameters such as sender, recipient, information type, purpose, and relevant relationships, together with logical conditions governing when a flow is permitted. These predicates can be represented as database relations, while their conjunctions, disjunctions, existential conditions, and negations can be expressed using relational operators and SQL. We can therefore compile a CI norm into an \emph{executable SQL view} whose result is exactly the database tuples satisfying the norm. Unlike asking an LLM whether a disclosure ``looks appropriate,'' the policy becomes an explicit, mechanically enforceable authorization boundary. The LLM is used only to \textit{translate} natural language into the SQL view, which can be inspected and verified at compile time, rather than making privacy decisions at runtime.


\noindent\textbf{Our contributions.} We propose \sysname{}, to the best of our knowledge, the first framework to provide formal privacy guarantees over natural-language text under \emph{either} CI or DP, enabling a new class of privacy-preserving text analytics. \sysname{} is \emph{flexible}: applications can request DP alone, CI alone, or their composition, each formally operationalized. 
\begin{packeditemize}
\item \textbf{Differential privacy-aware Text-to-SQL.}
We formulate a new Text-to-SQL task for finding equivalent lower-sensitivity SQL and develop a new search procedure. We construct the first benchmark for this task, where \sysname{} finds a correct, minimum-sensitivity formulation in \textbf{75.3\%} of cases, outperforming baselines by up to \textbf{45 percentage points}.

\item \textbf{Contextual integrity-aware Text-to-SQL.}
We introduce \textsc{Text-to-CQL}, a new task that grounds natural-language CI norms into SQL views. To the best of our knowledge, this is the first approach to operationalize CI as database views for natural-language analytics. On our new benchmark, current LLMs achieve up to \textbf{89\% F1}.

\item \textbf{A formal integration of CI and DP.}
We introduce \emph{contextual differential privacy}, a new privacy definition integrating CI and DP by \textbf{reshaping the neighborhood relation} of standard DP, and characterize when it yields strictly tighter sensitivity bounds. While prior work has argued for such an integration~\cite{rachel-cummings-dp}, to the best of our knowledge, we are the first to formalize it as a concrete privacy definition.

\item \textbf{End-to-end privacy over text.}
Together, these components enable a new paradigm for privacy-preserving text analytics: natural language remains the interface, relational databases provide the formal substrate, and privacy enforcement is explicit and auditable rather than delegated to LLM judgments.
\end{packeditemize}

\section{Background}\label{sec:background}



\noindent \textbf{Contextual Integrity (CI).} CI defines privacy by whether information flows appropriately within a particular context. For example, a hospital policy may state that ``a doctor may disclose a patient's diagnosis to a researcher for an approved study,'' while prohibiting the same disclosure to a budget analyst. We use the formalization of Barth et al.~\cite{barth2006privacy}, representing a norm as $N_i=(\lambda_i,r_i^{\mathrm{snd}},r_i^{\mathrm{rec}}, r_i^{\mathrm{sub}},\tau_i,\pi_i,\theta_i)$, where $\lambda_i\in\{+,-\}$ indicates whether the flow is permitted or prohibited, $r_i^{\mathrm{snd}}$, $r_i^{\mathrm{rec}}$, and $r_i^{\mathrm{sub}}$ are the sender, recipient, and subject roles (doctor, researcher, and patient), $\tau_i$ is the data type (diagnosis), $\pi_i$ is the purpose (approved study), and $\theta_i$ specifies conditions on the flow. A flow is permitted if it matches a positive norm and no applicable negative norm prohibits it.

\noindent \textbf{Differential Privacy (DP).}
DP is the de facto standard for protecting individual privacy in aggregate queries.

\begin{definition}
Let $I\sim I'$ denote databases differing by one tuple in a private relation.
A randomized mechanism $\mathcal{M}$ is $(\varepsilon,\delta)$-differentially
private if, for every output set $O$ and $I\sim I'$,
\[
\Pr[\mathcal{M}(I)\in O]
\le e^\varepsilon\Pr[\mathcal{M}(I')\in O]+\delta.
\]
\end{definition}

A standard way to achieve DP is to add noise calibrated to query \emph{sensitivity}. For example, the Laplace mechanism releases $q(I)+\eta$, where $\eta\sim Lap(\Delta q/\varepsilon)$ and $\Delta q$ bounds the change in $q$ between neighboring databases. Global sensitivity, $\GS{q}=\sup_{I\sim I'}|q(I)-q(I')|$, takes the worst case over all database instances and can be prohibitively large for queries with joins. In contrast, local sensitivity $\LS{q}(I)=\sup_{I'\sim I}|q(I)-q(I')|$ considers only neighbors of the current instance, yielding a tighter bound, but cannot be used directly because its data-dependent noise scale may itself leak information. Smooth sensitivity~\cite{nissim2007smooth} addresses this with an instance-dependent bound that varies smoothly across neighbors, but is intractable to compute exactly for general multiway joins. We therefore use \emph{residual sensitivity} (RS)~\cite{dong2021residual}, an efficiently computable, order-optimal smooth upper bound for this query class, to calibrate noise throughout the paper.


\noindent\textbf{Residual sensitivity (RS).} RS is the tightest known efficiently computable smooth sensitivity bound for multiway join counting queries. It captures how a tuple change propagates through joins using data-dependent join multiplicities. Selections can be pushed to their corresponding base relations before computing RS, while projections do not affect the bound. Fig.~\ref{fig:rs} shows the intuition behind RS. Starting from the full join
$q$, we remove one relation, $Z$, to form the \emph{residual query} $q_E$.
The \emph{boundary} is where this residual query connects back to the removed
relation. Here, the boundary is $B$. The value $b_1$ appears three times in
$q_E$, so a tuple added to $Z$ with $B=b_1$ can join with three tuples.
This multiplicity of $3$ captures how much one tuple change can affect the
query answer.

Formally, for a counting query $q=\bowtie_{i=1}^{n}R_i$, RS considers a
residual query $q_E$ which joins a subset of relations
$E\subseteq[n]$. Its boundary $\boundary{E}$ contains the attributes
connecting $q_E$ to the remaining relations. For each private relation
$R_i$, let $E=[n]\setminus\{i\}$. For each boundary value $b$,
$\multiplicity{i}{b}$ is the number of tuples in $q_{E}$ matching $b$.
RS takes the largest such multiplicity:
\begin{equation}
T_{E_i}(I)
=
\max_{b\in\operatorname{dom}(\boundary{E})}
\multiplicity{i}{b}.
\label{eq:te}
\end{equation}

We denote this maximum by $\witness{i}$, the largest number of output tuples
one tuple in $R_i$ can affect. Thus, $\LS{i}=\witness{i}$ and
$\LS{q}(I)=\max_{i\in\PrivRels}\LS{i}$. Because local sensitivity can vary
sharply across neighbors, RS also bounds how these multiplicities grow with
database distance. Let $\mathbf{s}=(s_1,\ldots,s_{|\PrivRels|})$, where $s_i$
counts tuple changes to private relation $R_i$, and let $\DistSet{k}$ contain
all such vectors summing to $k$. RS computes
\begin{gather}
T_{E,\mathbf{s}}(I)
=
\sum_{A\subseteq E\cap\PrivRels}
T_{E\setminus A}(I)\prod_{j\in A}s_j,
\label{eq:te-distance-vector}\\
{\mathrm{LS}}_{q,\mathbf{s}}(I)
=
\max_{i\in\PrivRels}T_{E_i,\mathbf{s}}(I),
\label{eq:ls-distance-vector}\\
{\mathrm{LS}}_q^{(k)}(I)
=
\max_{\mathbf{s}\in\DistSet{k}}
{\mathrm{LS}}_{q,\mathbf{s}}(I),
\label{eq:ls-distance}\\
\text{Finally,} \;\RS{q}(I)
=
\max_{k\ge0}e^{-\beta k}
\min\{{\mathrm{GS}}_q,{\mathrm{LS}}_q^{(k)}(I)\},
\label{eq:residual-sensitivity}
\end{gather}
where $\beta$ is the smoothing parameter. We denote the maximizing distance
and vector by $k^\star$ and $\mathbf{s}^\star$.

\begin{figure}
    \centering
    \includegraphics[width=1\linewidth]{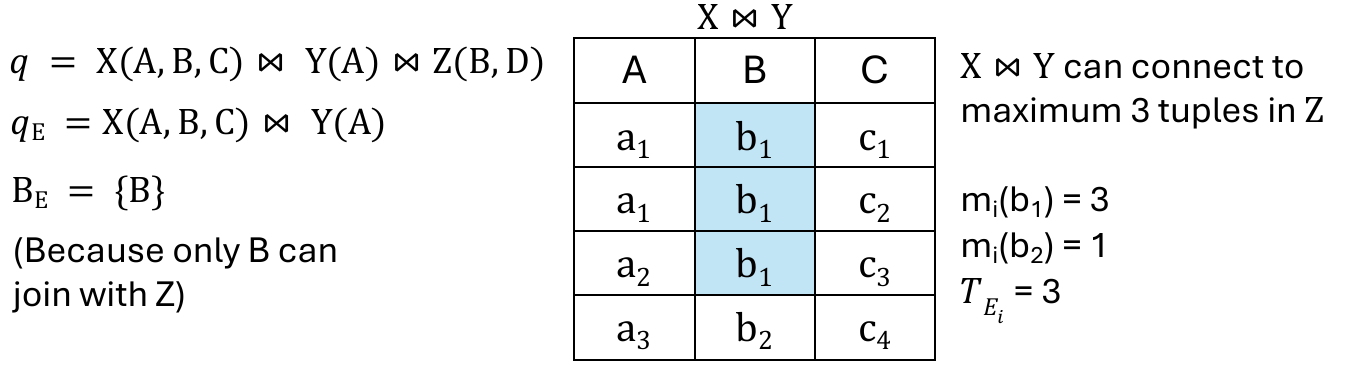}
    \caption{\textbf{RS.} Residual query, boundary, and multiplicity.}
    \label{fig:rs}
\end{figure}

\noindent\textbf{Key properties.}
We rely on three properties of RS:
\begin{packeditemize2}
\item[\textbf{(P1)}] \textbf{Query dependence.}
RS is defined over the complete query; adding a relation can increase or
decrease its sensitivity.

\item[\textbf{(P2)}] \textbf{Maximizer dependence.}
The final RS depends on the maximizing private relation, distance $k^\star$,
and distance vector $\mathbf{s}^\star$.

\item[\textbf{(P3)}] \textbf{Reusable intermediates.}
Computing RS produces boundary multiplicities $T_E$ for subsets of relations. We later reuse these in our search procedure.
\end{packeditemize2}
\begin{figure*}[t]
    \centering
    \includegraphics[width=1\linewidth]{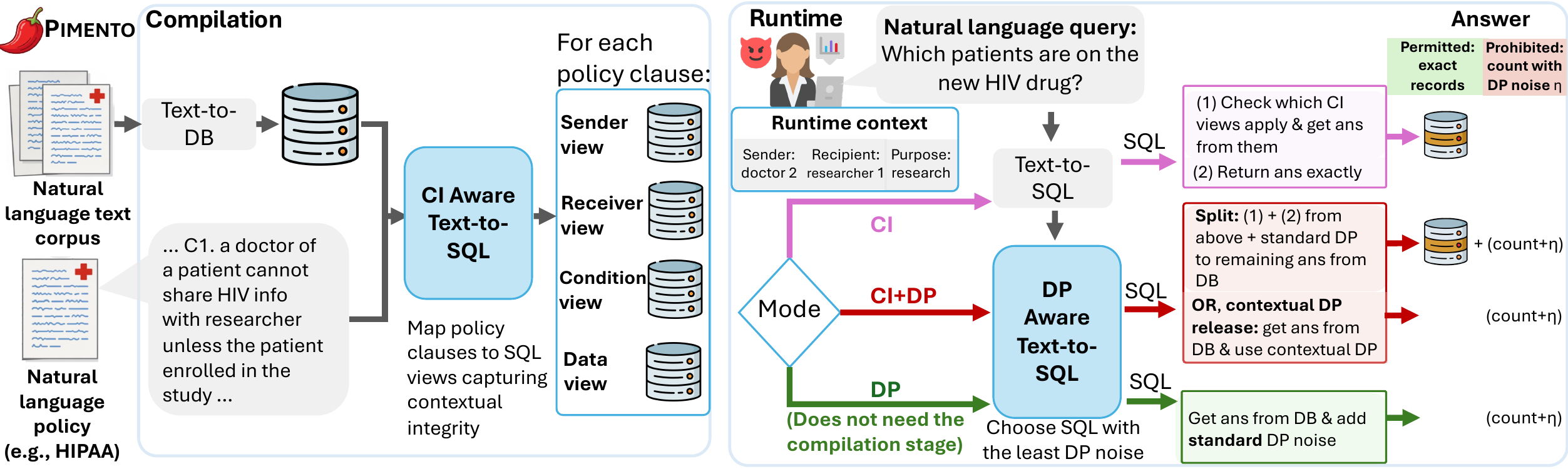}
   \caption{\footnotesize\textbf{\sysname{} overview}. \sysname{} takes three natural-language inputs: a text corpus, privacy policy, and query. During policy compilation, it (i) converts the text into a relational database and (ii) compiles the policy into SQL views representing contextual integrity norms using our \textbf{CI-aware Text-to-SQL} method. At runtime, given the query and context (sender, recipient, and purpose), \sysname{} supports three privacy-preserving release modes: (i) \emph{CI}, which releases authorized data exactly; (ii) \emph{CI+DP}, which either releases authorized data exactly and unauthorized data as a DP noisy count (\textit{split release}), or returns a noisy count over the full database using \textbf{contextual DP}, treating CI-authorized data as public; and (iii) \emph{DP}, which returns a noisy count with standard DP. CI modes use \textbf{CI-aware Text-to-SQL}, and DP modes use \textbf{DP-aware Text-to-SQL}. Our main contributions are the \textbf{CI-aware} and \textbf{DP-aware Text-to-SQL} mechanisms and the \textbf{definition of contextual DP}; Text-to-DB and standard Text-to-SQL are orthogonal components from prior work.}
    \label{fig:system}
\end{figure*}
\noindent \textbf{Text-to-SQL.} Recent state-of-the-art tree-based approaches to Text-to-SQL use Monte Carlo Tree Search (MCTS) to explore alternative SQL formulations during generation~\cite{li2025alpha,yuan2026mctssql}. MCTS balances \emph{exploitation}, i.e., revisiting promising regions of the search space, with \emph{exploration}, i.e., trying less-explored alternatives. In the tree, a node ($v$) corresponds to a partial or completed SQL query, and an edge ($a$) is an LLM reasoning action that produces a child node. Starting from a root node, MCTS tests whether the current node is fully explored (i.e., cannot have more children). If so, a child is chosen using a \textbf{selection rule}. Otherwise, an LLM performs \textbf{expansion}, proposing a batch of candidate children, from which one is selected randomly. This continues until a terminal node (a complete SQL query) is reached. The query is then evaluated, and a reward $Q(v,a)$ is assigned to the terminal node, typically based on SQL execution agreement. The reward is also backpropagated to every node along that query's branch, while updating the corresponding node and action visit counts $N(v)$ and $N(v,a)$. This constitutes one rollout, and rollouts continue until the budget is exhausted, after which we select a final query with the highest reward.

The selection rule is typically UCT~\cite{kocsis2006bandit}:

{\centering
$a^*=\arg\max_a\left[
\underbrace{\frac{Q(v,a)}{N(v,a)}}_{\text{exploitation}}
+c\underbrace{\sqrt{\frac{\log N(v)}{N(v,a)}}}_{\text{exploration}}
\right].$\par}

\noindent \textbf{Shapley Attribution.}
A game theoretic method~\cite{wikipedia_shapley_value} that rewards each decision by its contribution to the outcome.

\section{Problem Statement}\label{sec:setting}

\noindent \textbf{Motivating Scenario.} Our target use case is an organization holding volumes of sensitive unstructured text about individuals. The organization is the data custodian: it controls the corpus, defines its privacy policy, and mediates every query. Analysts may be internal or external, with different roles and access levels. The goal is to answer their analytical questions while enforcing policy and protecting the individuals described in the text.

We use a hospital as our running example. Much of its clinically valuable information is stored in free-form text, including physician notes, admission reports, medication records, and discharge summaries, and queried by diverse stakeholders. Natural language appears in three forms: (i) organizational \emph{data}; (ii) users' analytical \emph{questions}; and (iii) the \emph{privacy policy} governing access, such as HIPAA. To enable querying with formal privacy enforcement, we retain natural language at the interface but compute over a relational database. Specifically, we assume a Text-to-DB system has transformed the raw corpus into a database $D$ with schema $S$~\cite{squid,arora2025languagemodelsenablesimple,jiao-etal-2024-text2db}. We use the \textbf{schema below} throughout.

This setting raises three novel tasks, which we define next:
\emph{(i) Differential-privacy-aware Text-to-SQL}, generating SQL queries
optimized for low DP noise; \emph{(ii) Contextual-integrity-aware Text-to-SQL},
grounding CI policies into the database for enforcement; and
\emph{(iii) Integrating CI and DP}, composing both guarantees into end-to-end
release mechanisms.
\begin{center}
\begin{mybox2}{}
\scriptsize
\setlength{\tabcolsep}{0pt}
\renewcommand{\arraystretch}{1.35}
\begin{tabularx}{\linewidth}{@{}>{\itshape}l X@{}}
Clinical &
$\mathsf{Patient}(\mathsf{pat\_id},\ldots)$ \quad
$\mathsf{Visit}(\mathsf{pat\_id},\mathsf{diagnosis},\ldots)$ \newline
$\mathsf{MedicationOrders}(\mathsf{pat\_id},\mathsf{medication},\ldots)$ \quad
$\mathsf{Prescription}(\mathsf{pat\_id},\mathsf{medication},\ldots)$ \newline
$\mathsf{Discharge}(\mathsf{pat\_id},\mathsf{finaldiagnosis})$ \quad
$\mathsf{DxCategory}(\mathsf{diagnosis},\mathsf{hivstatus},\ldots)$
\\[4pt]
Roles &
$\mathsf{Employee}(\mathsf{employee\_id},\ldots)$ \newline
$\mathsf{TreatingRelationship}(\mathsf{employee\_id},\mathsf{pat\_id},\mathsf{role})$ \newline
$\mathsf{Supervises}(\mathsf{supervisor\_id},\mathsf{supervisee\_id})$ \newline
$\mathsf{CareTeamMember}(\mathsf{employee\_id},\mathsf{team\_id})$ \newline
$\mathsf{TeamAssignment}(\mathsf{team\_id},\mathsf{pat\_id})$
\\[4pt]
Research &
$\mathsf{Researcher}(\mathsf{researcher\_id},\ldots)$ \quad
$\mathsf{Study}(\mathsf{study\_id},\mathsf{researcher\_id})$ \newline
$\mathsf{StudyPatient}(\mathsf{study\_id},\mathsf{pat\_id})$
\\[4pt]
Constraint &
$\;\Gamma:\mathsf{DxCategory.diagnosis}\rightarrow
\mathsf{DxCategory.hivstatus}$
\end{tabularx}
\end{mybox2}
\captionof{table}{\footnotesize Schema for the running healthcare example. We assume standard primary and foreign key constraints.}
\label{tab:e2e-schema}
\end{center}

\subsection{Differential Privacy Aware Text-to-SQL}
Since user queries are in natural language, answering them over the relational database requires a Text-to-SQL model that maps the question to a SQL query. Standard Text-to-SQL maps a natural-language question $\NLQ$ and schema $S$ to a SQL query $q=\TSQL(S,\NLQ)$, whose answer $q(D)$ is released directly. We consider a different setting: the generated query is executed over a private database and its answer is released under differential privacy,
\begin{center}
   $q=\TSQL(S,\NLQ), \qquad
\hat a=q(D)+\eta,\quad
\eta\sim Lap\!\left(\frac{\RS{q}(D)}{\epsilon}\right).$ 
\end{center}

This seemingly simple change introduces two challenges. First, query generation
becomes part of the privacy mechanism: if the choice of $q$ depends on the
private database, it must itself satisfy DP. Second, utility depends on the
chosen query's sensitivity $\RS{q}(D)$, which determines the noise added to
its answer. Existing DP query-processing techniques sidestep the first
challenge by assuming $q$ is fixed and address the second \textit{post-hoc},
by reducing its sensitivity or release
noise~\cite{johnson2018towards,wilson2019differentially,kotsogiannis2019privatesql,bater2018shrinkwrap,dong2021residual}.
Text-to-SQL, however, presents a new opportunity: the query itself is
\textit{not} fixed.

\noindent\textbf{Ambiguity as Opportunity.}
We turn an inherent challenge of Text-to-SQL, ambiguity, into an opportunity
for improving privacy utility. While ambiguity about the user's intent is a
flaw, having multiple correct ways to express the same intent is a feature.
A natural-language question can admit multiple semantically correct SQL
formulations because schemas often represent the same relationship through
different tables, foreign keys, or join paths~\cite{amb1,amb2,amb3,amb4}.
Conventional Text-to-SQL has little reason to prefer among correct
formulations. Under DP, however, their differences can yield
substantially different sensitivities, and thus different noise.

Consider the query ``How many patients were prescribed drug $X$ by
cardiologists?'' The database supports two correct paths:
$\texttt{Prescription} \Join \texttt{Employee}$ and
$\texttt{MedicationOrder} \Join \texttt{Visit} \Join \texttt{Employee}$.
Both return the same answer, but can have very different sensitivities:
the first may associate one cardiologist with hundreds of prescriptions,
whereas the second may associate each visit with only a few medication orders.
Thus, different join paths induce different multiplicities and, consequently,
different sensitivity. Note that choosing the lowest-sensitivity formulation is \textbf{database-dependent}: neither path here is inherently better, and the question and schema alone cannot determine which has lower sensitivity.

These observations motivate a new problem that we call \textbf{DP-aware Text-to-SQL}: given a natural-language question $\NLQ$, schema $S$, and private database $D$, the goal is to select, among semantically correct SQL formulations, a query with low sensitivity while ensuring that the selection itself preserves differential privacy. In other words, we ask: \begin{quote} \vspace{-2mm}\textit{Can Text-to-SQL privately select the correct query that is cheapest to release under DP?}
\end{quote}

\subsection{Contextual Integrity Aware Text-to-SQL}\label{sec:setting:ci}
For CI, we are given a natural-language policy $\Pol$, the database constructed
from the raw text, and a user's query translated into SQL. The goal is to
answer the query while respecting the information flows permitted by $\Pol$.
Prior work commonly hands CI norms to an LLM, and asks it to decide what may be
revealed. This is fundamentally limited. First, authorization remains a model
judgment with no formal guarantee~\cite{pmlr-v267-shvartzshnaider25a,bagdasarian2024airgapagent}.
Second, LLMs make substantial contextual privacy errors: ConfAIde reports
private information disclosure in 57\% of cases~\cite{mireshghallah2024can}.
Third, placing the privacy decision inside the LLM exposes enforcement to
prompt-injection attacks~\cite{bagdasarian2024airgapagent,zhan2024injecagent}.
Our key observation is that CI norms are logical rules and should therefore
define the authorization boundary, rather than merely guide an LLM. We use the
LLM only to \textit{translate} natural-language norms into executable rules
over the schema; the final privacy decision is deterministic.

\noindent\textbf{Shared Representation as Opportunity.}
To make this translation possible, we make a key observation: CI norms and SQL share a common logical foundation. A CI
norm specifies who may access what information under
what conditions, which can be expressed as predicates over database
relations and attributes. Thus, once grounded in the schema, a CI norm can be
compiled directly into executable SQL.

\noindent\textbf{Proposed Workflow.}
Consider a hospital policy: an attending physician may disclose a patient's
diagnosis to a resident for treatment, only if the attending treats the patient
and supervises the resident. We can compile this clause into a CI norm, represented
by SQL views: a sender view
identifies attendings from \textsf{TreatingRelationship}, a recipient view
identifies residents from \textsf{Employee}, a data view locates diagnosis
information in \textsf{Visit}, and condition views encode the required
treatment and supervision relationships using join queries between these tables. Hereafter, a \emph{CI norm} refers
to this collection of SQL views representing its roles, data, purpose, and
conditions.

At runtime, these views determine which norms apply and the data they
authorize. Suppose resident
$20$ queries diagnosis data from physician $10$. We can identify norms whose data views cover the
queried information, bind physician $10$ and resident $20$ to their views, and execute them over the database. The
result contains only diagnoses satisfying the norm's conditions. Repeating
this for all applicable norms, and subtracting prohibited
data from permitted data, enables releasing the resulting authorized
data exactly.

This motivates the problem we call \textbf{contextual integrity aware Text to SQL}, or \textsc{Text-to-CQL} in short. \vspace{-2mm}
\begin{quote}
\textit{\textsc{Text-to-CQL}: Given a natural-language privacy policy $\Pol$, relational schema $\Sch$ and
a user query, how can we compile the policy into executable SQL views that,
instantiated with the query context at runtime, enforce exactly the information the
query is authorized to release under contextual integrity?} 
\end{quote}

\subsection{Integrating Contextual Integrity and Differential Privacy}


So far, we have defined tasks for enforcing CI and DP independently in our setting. Yet integrating the two can provide benefits that neither offers alone. Contextual integrity is \textit{normative}: it determines whether an information flow is appropriate. Differential privacy is \textit{descriptive}: it limits what can be learned about individuals from that flow through calibrated noise. both \emph{what} information may flow and \emph{how much} that flow may reveal. There has been growing interest in bringing these two perspectives together. Benthall and Cummings~\cite{rachel-cummings-dp} explicitly advocate integrating CI and DP, while related work on policy-aware DP shows that public or policy-specified information can be incorporated into sensitivity analysis~\cite{dp4sql,blowfish,kifer2014pufferfish}. What has been missing is an operational bridge: CI norms have traditionally remained logical specifications rather than executable objects that a DP mechanism can directly consume. \sysname{} provides this bridge. By compiling natural-language CI norms into executable SQL views (Sec.~\ref{sec:setting:ci}), we obtain a concrete, database-grounded authorization boundary that can be composed directly with DP.

Beyond making CI and DP independently composable within our end to end setting, we identify a further opportunity from their integration: improving DP utility. Standard DP protects all records uniformly, even when CI already authorizes some information to flow to the querier. But if that information is already appropriate to reveal, should it still be used for noise calibration? Our key insight here is to \textbf{revisit the notion of DP neighborhood} under contextual integrity. Since we can identify the records authorized by CI at runtime, we can treat those records as public and hold them fixed in the neighborhood relation. This can reduce local sensitivity by excluding changes to already authorized information. The remaining question is whether this reduction also yields a tighter smooth upper bound through residual sensitivity.

This leads to our third problem:
\begin{quote}\vspace{-2mm}
\textit{Can contextual integrity enable tighter smooth sensitivity bounds for differentially private query answering?}
\end{quote}

\section{\sysname{}: System Description}
\label{sec:name}

\sysname{} is a privacy framework for natural-language analytics over
unstructured text. To the best of our knowledge, \sysname{} is the first
framework to provide formal privacy guarantees for natural-language text under
either CI or DP, and to support their combination within a unified framework.

As motivated in Sec.~\ref{sec:setting}, we envision \sysname{} being deployed
by the data custodian, e.g., the hospital. It takes three natural-language
inputs: (i) the text corpus, (ii) a user query, and (iii) a privacy policy,
and enforces the requested privacy semantics before releasing an answer. It is
\emph{flexible}: applications may request DP alone, CI alone, or their
composition, and \sysname{} formally operationalizes each choice.

\sysname{} realizes these guarantees by grounding all three inputs in a common
relational database representation. The text corpus is converted into a
relational database using existing Text-to-Database techniques, while the
user's question is translated into SQL using Text-to-SQL. Our DP modes support
single and multiway join counting queries with selections and projections,
while CI alone supports arbitrary SQL queries. Database construction is
orthogonal to \sysname{}: it is a well-studied problem, and \sysname{} can
build on any existing Text-to-Database
system~\cite{squid,arora2025languagemodelsenablesimple,jiao-etal-2024-text2db,liang2026beyond,lin2026structure}.
Recent work already achieves over 95\% performance on this task~\cite{squid},
so we focus on the privacy layer over the resulting representation.
We develop DP-aware Text-to-SQL to optimize query generation for private
release, CI-aware Text-to-SQL to compile natural-language CI policies into
enforceable database views, and mechanisms for integrating the resulting
authorization boundary with DP. We describe each component next, followed by
guidelines for choosing among these guarantees. Fig.~\ref{fig:system} provides
an overview.

\subsection{Differential Privacy Aware Text-to-SQL}
\label{sec:dp-text2sql}

We start with two strawman solutions to build intuition.
\begin{packeditemize}
\item \textbf{Strawman Solution 1.} A natural approach is to fine-tune
models to favor low-sensitivity queries. However, no SQL
is universally low-sensitivity: residual sensitivity is \emph{data-dependent}. Thus, a model trained
on one set of databases cannot reliably identify the lowest-sensitivity
formulation on a new database, as our experiments confirm
(Sec.~\ref{sec:eval-rq1}).

\item \textbf{Strawman Solution 2.} Another approach is to generate the
top-$k$ candidates, privately evaluate their sensitivities, and select
the minimum. This has two limitations. First, evaluating more candidates
consumes privacy budget because residual sensitivity is
data-dependent. Second, the lowest-sensitivity formulation may lie outside the
top-$k$, using alternative relations, join paths, or predicates that receive
low model probability. Our experiments confirm that top-$k$ candidates often
lack this structural diversity (Sec.~\ref{sec:eval-rq1}). At the other extreme, exhaustive
search covering all structures is intractable.
\end{packeditemize}
These limitations require sensitivity optimization at \textit{runtime}, using
database-dependent feedback to guide search. This in turn creates two challenges:
 \textbf{(1) Use of private data-dependent
feedback}, since the search uses private database-dependent feedback, and \textbf{(2) Correctness-preserving sensitivity optimization}, since reducing
sensitivity must not come at the cost of semantic correctness.
Tree search is a natural fit: unlike top-$k$ decoding, MCTS uses feedback from
explored queries to search beyond the model's initial distribution. \sysname{}
extends MCTS by making \emph{sensitivity a first-class search objective},
steering search toward semantically correct, low-sensitivity formulations
while preserving privacy. We realize this in four stages, illustrated in Fig.~\ref{fig:mcts}.
\begin{figure*}[t]
    \centering
    \includegraphics[width=0.9\linewidth]{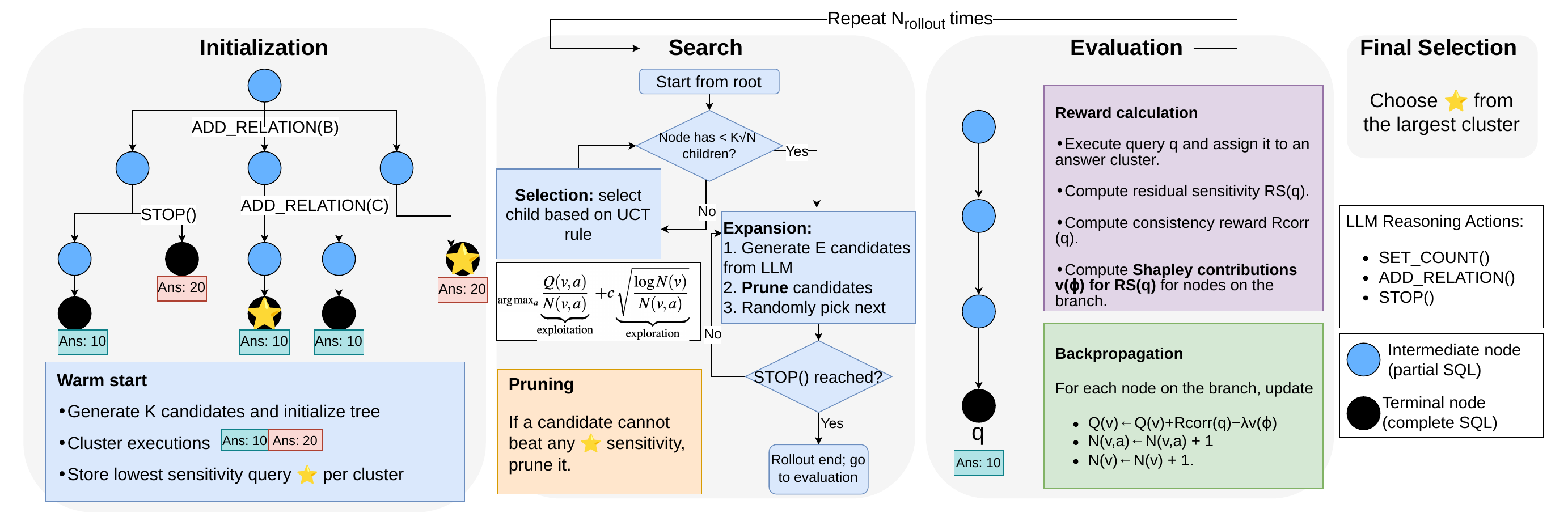}
    \caption{The four-stage Monte Carlo Tree Search procedure for \textbf{DP-aware Text-to-SQL} in \sysname{}.}
    \label{fig:mcts}
\end{figure*}

\subsubsection{Initialization}
\label{sec:dp-text2sql:notation}
At initialization, \sysname{} uses three techniques:

\noindent\textit{Private Synthetic Dataset.}
MCTS evaluates many candidates using two data-dependent signals:
(i) residual sensitivity for DP utility and (ii) result self-consistency
for semantic correctness. Repeatedly computing these on $\PrivDB$ would
accumulate privacy loss. Instead, \sysname{} pays for data dependence once by
generating a private synthetic database
$\SynDB \leftarrow
\SynMech(\PrivDB;\varepsilon_{\mathrm{syn}},\delta_{\mathrm{syn}})$
before search and using $\SynDB$ for all computations. This incurs
$(\varepsilon_{\mathrm{syn}},\delta_{\mathrm{syn}})$ once; all subsequent use
of $\SynDB$ is post-processing. We use PrivPetal~\cite{cai2025privpetal},
which preserves cross-table correlations in relational data.

\noindent\textit{DP-Compatible Search Tree.}
DP also constrains what MCTS may generate. Generic Text-to-SQL can introduce
arbitrary SQL operators, while our DP mechanism supports counting queries over
multiway joins with selections and projections (Sec.~\ref{sec:background}).
We therefore make the search DP-compatible by construction: each node is a
partial query, and each edge is restricted to one of three actions:
\begin{packeditemize}
\item $\ActSetCount(x)$ sets the aggregation to
$\gamma_{\operatorname{COUNT}(x)}$, where $x$ is $*$\footnote{
$\operatorname{COUNT}(*)$ returns the total number of result rows.} or an
attribute, with $*$ as the default.
\item $\ActAddRelation(R,J,F)$ adds relation $R$ with join predicate $J$ and
filter $F$. For the first relation, $E\leftarrow\sigma_F(R)$; otherwise,
$E\leftarrow\sigma_F(E\bowtie_JR)$.
\item $\ActStop()$ terminates construction and returns
\scalebox{0.8}{$q=\gamma_{\operatorname{COUNT}(x)}(E)$}.
\end{packeditemize}
\noindent\textit{Warm-start.}
Since zero shot Text-to-SQL is optimized for accuracy, we use \(K\) zero shot samples to warm start the tree and simplify the problem, decomposing each sample into our actions and inserting it as a tree path. We execute these queries on $\SynDB$ and group those with the same answer into \emph{execution clusters} $C_y$. Queries within a cluster agree on their answer but may differ structurally and in sensitivity. For each cluster, we maintain its lowest-sensitivity \emph{incumbent} $q_y$. Clusters and incumbents are updated as the search progresses.

\subsubsection{Search}
\label{sec:dp-text2sql:pruning}
As described in Sec.~\ref{sec:background}, MCTS alternates between selection and expansion to explore new query formulations. However, broad exploration for low-sensitivity formulations is computationally expensive, while naively expanding every branch quickly becomes intractable. \sysname{} therefore uses two complementary techniques to better focus the search:
\\\noindent\textit{Sensitivity-Guided Pruning.}
\sysname{} uses a greedy pruning strategy to eliminate branches whose next
action is unlikely to improve any formulation found so far. Consider a partial
query $q_v$ and a proposed equijoin to relation $R$ on $q_v.A=R.B$. We
consider the extension only when it is \emph{containment preserving} on
$\SynDB$, meaning every join value in $q_v$ has a match in $R$: $\pi_A(q_v(\SynDB))
    \subseteq
    \pi_B(R(\SynDB)).$
Although sensitivity can generally increase or decrease as relations are added
\textbf{(P1, Sec.~\ref{sec:background})}, containment preservation identifies
a special case: if $R$ retains every join value from $q_v$, adding $R$ cannot
lower sensitivity because it preserves the maximum join multiplicity. This is
a greedy condition, as the guarantee need not hold after subsequent actions.
Let $q_{v'}$ denote the extended query (with the proposed action) and
$\Incumbent{\max}=\max_y\Incumbent{y}$ the largest sensitivity incumbent among
the execution clusters found so far. We prune
when $\ResRS{q_{v'}} > \Incumbent{\max}.$
The proposed action has already raised sensitivity above every known cluster
incumbent, making the branch less promising. We delay pruning for a few
rollouts to allow the clusters and incumbents to stabilize.

\noindent\textit{Disallowing Duplicate Branches.}
Sibling nodes using the same relation key create redundant subtrees. We pass existing sibling keys as an exclusion set when generating children and discard any that reuse them.


\subsubsection{Evaluation}
\label{sec:dp-text2sql:correctness}
When a rollout reaches $\ActStop()$, it produces a completed query $q$, which
\sysname{} evaluates for correctness and sensitivity. We design two rewards:
\\\noindent\textbf{Correctness Reward.}
We estimate correctness using result self-consistency. After observing $M$ valid queries, we define
\begin{equation}
\scalebox{0.99}{$
    \RewCorr{q}
    =
    \frac{|C_{y(q)}|-1}{M-1},
    \qquad
    \RewCorr{q}=0 \text{ if } M\leq 1.
$}
\label{eq:consistency}
\end{equation}
Thus, queries supported by larger execution clusters receive higher reward.
\\\noindent\textbf{Sensitivity Reward.}
\label{sec:dp-text2sql:credit}
By \textbf{(P1)}, RS is defined only for a \emph{complete} query, whereas MCTS
constructs queries one relation at a time. Scoring partial queries is
insufficient because adding a relation can raise or lower sensitivity. A
terminal reward such as $1/(1+\log\ResRS{q})$ is also too coarse: it assigns every action the same
signal despite their potentially different effects. We also confirm that this performs poorly (Sec.~\ref{sec:eval-rq1}).

We instead treat sensitivity feedback as a relation-level credit-assignment
problem. Viewing the relations of $q$ as players in a cooperative game with
value $\ResRS{q}$, we use Shapley attribution~\cite{wikipedia_shapley_value} to
assign each relation a contribution $\phi_i$. These contributions sum to
$\ResRS{q}$ and are independent of construction order. We (i) define a value
function over relation subsets using quantities already computed by RS,
(ii) compute each relation's Shapley contribution, and (iii) normalize and
negate these contributions so sensitivity-increasing relations are penalized
and sensitivity-reducing ones rewarded. This converts query-level RS into an action-level signal for MCTS.
\smallskip
\\\noindent\textit{Sensitivity Value Function.}
Consider a completed query \scalebox{0.9}{$q=\bowtie_{i=1}^{n}R_i$} with relation indices
\scalebox{0.9}{$V=[n]$}. For distance vector $\mathbf{s}$,
\begin{equation*}
\scalebox{0.9}{$\ResRS{q}
=
\max_{k\geq 0}
e^{-\beta k}
\min\left\{
    {\mathrm{GS}}_q,\,
    \max_{\mathbf{s}\in\DistSet{k}}
    \max_{i\in\PrivRels}
    T_{V\setminus\{i\},\mathbf{s}}(\SynDB)
\right\}.$}
\label{eq:rs-full}
\end{equation*}
(Sec.~\ref{sec:background}). For each residual query $q_E$, RS computes
\begin{equation*}
\scalebox{1}{$
T_{E,\mathbf{s}}(\SynDB)
=
\sum_{A\subseteq E\cap\PrivRels}
T_{E\setminus A}(\SynDB)
\prod_{j\in A}s_j.
$}
\label{eq:rs-te-expansion}
\end{equation*}
Thus, terminal RS already computes maximum boundary multiplicities $T_E$ for
relation subsets \textbf{(P3)}. \sysname{} caches and reuses them
to attribute to individual relation choices.

By \textbf{(P2)}, the terminal RS depends on its maximizing distance $k^\star$ and distance vector $\mathbf{s}^\star$. We therefore retain these maximizing values from the terminal computation and, holding them fixed, define for each $U\subseteq V$
\begin{equation}
\scalebox{0.9}{$
W(U)
=
e^{-\beta k^\star}
\min\left\{
    {\mathrm{GS}}_q,\,
    \max_{i\in U\cap\PrivRels}
    T_{U\setminus\{i\},\mathbf{s}^\star}(\SynDB)
\right\}.
$}
\label{eq:rs-vector}
\end{equation}
where the maximum over an empty set is defined as zero. Hence
$W(\emptyset)=0$ and, because the terminal maximizing values are retained,
$W(V)=\ResRS{q}$.

Importantly, $W(U)$ is not the residual sensitivity of the subquery
containing only $U$. Instead, it measures the contribution of relations in
$U$ to the fixed terminal residual-sensitivity computation: we retain
$k^\star$, $\mathbf{s}^\star$, and ${\mathrm{GS}}_q$ rather than
re-optimizing residual sensitivity for every subset. We refer to
$\{W(U):U\subseteq V\}$ as the sensitivity vector.

\noindent\textit{Shapley attribution.} With the value function $W$, we divide the total sensitivity $W(V)$ among the relations. A relation can affect residual sensitivity by changing either the maximum boundary multiplicities, or which private relation attains the maximum in Eq.~\ref{eq:rs-vector}. Shapley attribution captures both \emph{independently of the order} in which MCTS constructs the query, so each relation's contribution reflects its intrinsic effect rather than the search trajectory. For relation $R_i$, its Shapley contribution is
\begin{equation}
\scalebox{0.9}{$
\phi_i
=
\sum_{U\subseteq V\setminus\{i\}}
\frac{|U|!(n-|U|-1)!}{n!}
\left[
    W(U\cup\{i\})-W(U)
\right].
$}
\label{eq:shapley-rs}
\end{equation}
Here $U$ is a subset of relations that could already be present before $R_i$ is
added, and $W(U\cup\{i\})-W(U)$ measures how much the terminal computation changes
when $R_i$ joins $U$. Because this effect depends on which relations are already
present, we average it over every $U\subseteq V\setminus\{i\}$. The weight
$\frac{|U|!(n-|U|-1)!}{n!}$ is the fraction of the $n!$ relation orderings in which
exactly $U$ precedes $R_i$ ($|U|!$ orders before, $(n-|U|-1)!$ after), so $\phi_i$
is $R_i$'s order-independent contribution to residual sensitivity.

Since $W(\emptyset)=0,$ $W(V)=\ResRS{q}$, Shapley efficiency gives
\begin{equation*}
\scalebox{0.9}{$
\sum_{i\in V}\phi_i
=
W(V)-W(\emptyset)
=
\ResRS{q}.
$}
\label{eq:shapley-efficiency}
\end{equation*}
Thus, the Shapley values exactly decompose the residual sensitivity of the
full query into relation-level contributions.

\smallskip
\noindent\textit{Example.}
Consider $q=x\Join y\Join z$, where all three relations are private. If a
tuple in $x$ changes, the remaining relations $\{y,z\}$ determine how much
that change can be amplified. At the distance vector which resulted in the maximum residual sensitivity of the final query
$\mathbf{s}^\star$, this bound is
\begin{equation*}
\scalebox{0.9}{$
T_{\{y,z\},\mathbf{s}^\star}(\SynDB)
=
T_{\{y,z\}}(\SynDB)
+s_y^\star T_{\{z\}}(\SynDB)
+s_z^\star T_{\{y\}}(\SynDB)
+s_y^\star s_z^\star T_{\emptyset}(\SynDB).
$}
\label{eq:rs-te-xyz}
\end{equation*}
Changing a tuple in $y$ or $z$ analogously gives
$T_{\{x,z\},\mathbf{s}^\star}(\SynDB)$ and
$T_{\{x,y\},\mathbf{s}^\star}(\SynDB)$, respectively. The maximum over
these cases determines the local-sensitivity term used by
$W(\{x,y,z\})$.

The Shapley contribution of $z$ then averages its marginal effect across
all possible contexts:
\begin{equation*}
\scalebox{0.9}{$
\begin{aligned}
\phi_z
={}&
\tfrac13[W(\{z\})-W(\emptyset)]
+\tfrac16[W(\{x,z\})-W(\{x\})] \\
&+\tfrac16[W(\{y,z\})-W(\{y\})]
+\tfrac13[W(\{x,y,z\})-W(\{x,y\})].
\end{aligned}
$}
\label{eq:shapley-z}
\end{equation*}
For example, $W(\{x,z\})-W(\{x\})$ measures the effect of adding $z$
when $x$ is already present. The coefficients account for how often each
term occurs across all possible relation orderings. Thus, $\phi_z<0$
means that adding $z$ lowers sensitivity on average, whereas $\phi_z>0$
means that it raises sensitivity.

\noindent\textit{Sensitivity credit.}
We normalize each Shapley contribution as
$\nu(x)=\operatorname{sgn}(x)\log(1+|x|)/\log(1+S_{\mathrm{ref}})$, where
$S_{\mathrm{ref}}$ is the median warm-start RS, compressing its magnitude
while preserving its sign. If action $a_t$ adds relation $R_i$, it receives
\begin{equation}
\scalebox{0.9}{$
R_{\mathrm{sens}}(i,q)
=
-\lambda_s\nu(\phi_i).
$}
\label{eq:sensitivity-credit}
\end{equation}
where $\lambda_s$ controls its weight.

\noindent\textbf{Backpropagation.} The two rewards are backpropagated at different granularities. Correctness is a query-level signal and is propagated along the entire trajectory. Sensitivity is relation-specific and is assigned only to the edge that introduced the corresponding relation.
Thus, for an action $a_t$ we update:
\begin{equation}
\scalebox{0.9}{$
\begin{aligned}
Q(v_{t-1},a_t)
&\leftarrow
Q(v_{t-1},a_t)
+
\RewCorr{q}
-
\lambda_s\nu(\phi_i),\\
N(v_{t-1},a_t)
&\leftarrow
N(v_{t-1},a_t)+1,\\
N(v_{t-1})
&\leftarrow
N(v_{t-1})+1.
\end{aligned}
$}
\label{eq:structured-backup}
\end{equation}
For actions that do not add a relation, only $\RewCorr$ is
applied. Thus, every action receives credit for correctness, while relation-adding actions also get credit for sensitivity reduction, letting MCTS learn which choices yield lower sensitivity.
\subsubsection{Final Selection}
\label{sec:dp-text2sql:finalselect}

At the end of search, we select the lowest sensitivity incumbent query from the largest cluster $
    y^\star=\arg\max_y |C_y|,
    \qquad
    \OptQ=q_{y^\star}.$ We then compute $\RS{\OptQ}(\PrivDB)$ on the private database to calibrate
noise and release the answer. Thus, by choosing the largest cluster first, correctness takes priority over
low sensitivity. Algorithm~\ref{alg:dp-mcts} in App.~\ref{app:alg} summarizes
the complete search.



\begin{theorem}[DP of Text-to-SQL]
\label{thm:end-to-end-dp}
If the synthetic release is \scalebox{0.8}{$(\varepsilon_{\mathrm{syn}},\delta_{\mathrm{syn}})$}-DP and the final query release is \scalebox{0.8}{$(\varepsilon_{\mathrm{rel}},\delta_{\mathrm{rel}})$}-DP, then \sysname{} is
\scalebox{0.8}{$(\varepsilon_{\mathrm{syn}}+\varepsilon_{\mathrm{rel}},
\delta_{\mathrm{syn}}+\delta_{\mathrm{rel}})$}-DP (proof in App.~\ref{app:proofs}).
\end{theorem}


\subsection{Contextual Integrity Aware Text-to-SQL}
\label{sec:citosql}
\sysname{} takes a natural-language privacy policy $\Pol$ and database schema
$\Sch$ and compiles the policy into SQL views encoding what information may be
disclosed, to whom, and under what conditions. We call this task
\textsc{Text-to-CQL}. At runtime, \sysname{}
instantiates these views with the query context to derive $\Vex$, the exact
information authorized for disclosure. This design removes the LLM from the privacy decision: the LLM only translates policy text into SQL, while authorization is enforced symbolically over the
compiled views. This avoids relying on LLM judgment for privacy decisions and
enables systematic validation of its output, including detecting compilation
errors when a policy clause cannot be represented in the database.

This introduces three challenges: \textbf{(1) Grounding symbolic norms},
mapping policy concepts such as roles, information, and conditions to
concrete database operations; \textbf{(2) Resolving interactions across norms}, ensuring every policy clause is represented (coverage), norms do not conflict or duplicate one another
(satisfiability and non-redundancy), and permitted information does not imply
prohibited information (information implication);
and \textbf{(3) Enforcing context at runtime}, determining authorization for
each query based on its context and database state. \sysname{} addresses these
through three steps: \emph{policy-to-view translation} compiles
CI norms into SQL views, \emph{view canonicalization} validates and resolves
interactions among them, and \emph{runtime enforcement} instantiates them to
derive the authorized information $\public$. The first two run once per policy
and schema, while enforcement runs per query. We consider the following running policy:
\vspace{-1.5mm}
\begin{quote}
An attending physician may disclose a patient's diagnosis to a resident
for treatment if the resident is supervised by the attending and is on the
patient's care team.
\end{quote}

\subsubsection{Policy-to-View Translation.}
The core of \textsc{Text-to-CQL} is translating each natural-language policy
clause into executable database views. \sysname{} first extracts a structured CI norm $N_i$ from each policy clause,
following prior work~\cite{shvartzshanider2023beyond,fan2024goldcoin,li2025privacy},
and then grounds each component in $\Sch$ using an LLM. We provide the norm,
source clause, schema, and examples for three types of views. \emph{Role views}
($\Vsnd{},\Vrec{},\Vsub{}$) identify valid sender, recipient, and subject
bindings; in our example, they map the attending and resident roles to
\textsf{Employee}. \emph{Condition views} ($\Vth{}$) encode the relationships
required by the policy. These may require a lookup, such as the attending
relationship in \textsf{TreatingRelationship}; a join, such as care-team
membership through \textsf{CareTeamMember} and \textsf{TeamAssignment}; or
recursion, such as indirect supervision through \textsf{Supervises}. Finally, \emph{data views} ($\Vdat{}$)
identify all locations containing the governed information. For example,
$\mathsf{Diagnosis}$ may appear in both \textsf{DxCategory} and
\textsf{Discharge}; \sysname{} unions these locations and tags each value with
its provenance $o=(R,\mathit{tid},A)$, identifying the source relation, tuple,
and attribute. This provenance is later used for the neighborhood relation in
Sec.~\ref{sec:ci-dp}. Therefore, for each $N_i$, translation produces
$\Vset_i=\{\Vsnd{i},\Vrec{i},\Vsub{i},\Vth{i},\Vdat{i}\}$, which jointly
captures the CI constraints.

\subsubsection{View Canonicalization and Verification}
\label{sec:stage-policy}
Compiling each norm independently can violate \emph{coverage},
\emph{satisfiability}, and \emph{non-redundancy}, while missing
\emph{information implications} across norms. Our key insight is that grounding
the norms in the database enables symbolic checks for these violations. We
perform four checks over the compiled views.

\textbf{(1) Coverage:} we assign each policy clause an identifier and ensure
that it produces at least one compiled view. If not, \sysname{} reports the
clause as uncompilable and returns a policy compilation error.
\textbf{(2) Redundancy:} \sysname{} executes the compiled views over applicable
database records and compares their outputs. Views governing the same
information are merged, yielding a canonical policy representation.
\textbf{(3) Satisfiability:} individually valid norms may collectively produce
contradictory permissions and prohibitions. \sysname{} therefore evaluates
them together and checks that the policy admits at least one valid information
flow; otherwise, it reports the policy as unsatisfiable.
\textbf{(4) Information implication:} \sysname{} identifies when information
governed by one view reveals information governed by another and materializes
this dependency into the policy. For example, if
$\mathsf{Diagnosis}\rightarrow\mathsf{HIVStatus}$ and HIV status is prohibited,
\sysname{} adds a prohibition over diagnosis records,
ensuring they are also excluded at runtime. Thus, the result is a consistent set
of policy views ready for runtime enforcement.

\subsubsection{Runtime Policy Enforcement}

The compiled views still contain unknown parameters (e.g., who exactly is the querier?), and must be contextualized for each query. We therefore require each query to be also be supplied  with the runtime context:
the sender $s_\rho$, querier $r_\rho$, declared purpose $\pi_\rho$, and SQL
query $\ReqQuery$, representing a query as
$\Req=(s_\rho,r_\rho,\pi_\rho,\ReqQuery)$. Enforcement proceeds in three steps.

\noindent\textbf{(1) Identify relevant norms.}
\sysname{} first identifies the attributes returned by $\ReqQuery$. For
example, if the query queries patient IDs and diagnoses from
\textsf{Visit}, the returned attributes are $\mathsf{Visit.pat\_id}$ and
$\mathsf{Visit.Diagnosis}$. It then selects norms whose data views $\Vdat{i}$
govern any of these attributes, restricting enforcement to relevant norms.

\noindent\textbf{(2) Instantiate and evaluate norms.}
For each relevant norm, \sysname{} substitutes the query context
$(s_\rho,r_\rho,\pi_\rho)$ into its views and evaluates them over the current
database. Suppose attending physician $10$ queries diagnosis information for
resident $20$ for \texttt{treatment}. \sysname{} binds physician $10$ as the
sender and resident $20$ as the recipient, then evaluates the role and
condition views to identify patients for whom the policy conditions hold.
The data view then identifies the diagnosis values governed by the norm.
Repeating this for all relevant positive and negative norms produces the
permitted information $\Pos_\rho$ and prohibited information $\Viol_\rho$,
giving the final authorized set
$\Vex=\Pos_\rho\setminus\Viol_\rho$.

\noindent\textbf{(3) Construct the authorization boundary.}
$\Vex$ gives exactly the data authorized for release and can directly
answer the query under CI. For integration with DP in the next section, we
also identify the corresponding database cells. Since each data view retains
the origin $o=(R,\mathit{tid},A)$ of every value, \sysname{} maps $\Vex$ back
to its source cells to form $\public$.

\begin{theorem}[Equivalence to contextual integrity]
\label{thm:ci-equivalence}
Assuming the CI norms are faithfully compiled by the LLM, the resulting SQL
views are equivalent to contextual integrity: for any database and query,
they authorize exactly the information flows permitted by CI (proof is in App.~\ref{app:proofs}).
\end{theorem}

\subsection{Integrating Contextual Integrity and Differential Privacy}
\label{sec:ci-dp}
The semantics of integrating CI and DP is as follows. We wish to release a query's
answer under DP, but we now have additional contextual information about the query:
the CI policy specifies which information flows are authorized and any record
authorized to flow is effectively public to the querier. Intuitively, this should reduce the
sensitivity that DP must protect against. The key lies in the neighborhood:
standard DP protects against all neighboring databases; CI rules out neighbors
that differ only in information already authorized to the querier. This yields
a finer-grained neighborhood that holds contextually public information fixed
while continuing to protect what remains private. Unlike prior approaches that
assume coarse public information such as relation cardinalities~\cite{dp4sql}, CI can
expose richer structure, including individual tuples, attributes, and
relationships. We formalize this as \emph{contextual differential privacy} and
characterize exactly when it yields tighter sensitivity bounds than standard DP.

\begin{definition}[Contextual Differential Privacy]
\label{def:contextual-dp} Fix a query context
$c=(s_\rho,r_\rho,\pi_\rho,q_\rho)$.
Let $\mathcal{A}_c(D)$ denote the information in $D$ that the
CI policy authorizes to flow under context $c$.
Define the contextual neighboring relation
\[
D \sim_c D'
\iff
D \sim D'
\;\land\;
\public=\mathcal{A}_c(D'),
\]
where $D\sim D'$ denotes the standard neighboring relation.

A randomized mechanism $\mathcal{M}$ satisfies
$(\epsilon,\delta)$-\emph{contextual differential privacy} for context $c$
if, for all $D\sim_c D'$ and all measurable sets of outputs $O$,
\[
\Pr[\mathcal{M}(D,c)\in O]
\le
e^\epsilon
\Pr[\mathcal{M}(D',c)\in O]
+\delta.
\]
\end{definition}

\noindent\textbf{Example.}
Suppose a resident queries diagnosis
data for treatment. The database contains patients
$\texttt{P1}$, $\texttt{P2}$, and $\texttt{P3}$, whose removal changes the
query answer by $10$, $6$, and $2$, respectively. Under add/remove
DP, the worst-case change is $10$. Now suppose CI authorizes the
diagnoses of $\texttt{P1}$ and $\texttt{P2}$, placing them in $\public$.
Contextual neighborhood must preserve $\public$, leaving only $\texttt{P3}$ as a
valid removal and reducing the worst-case change to $2$. However, insertion of
a new unauthorized record remains possible: if a new $\texttt{P4}$ changes the
answer by $12$, the insertion still preserves $\public$. Thus, CI may restrict
removals without restricting insertions, unless schema constraints limit insertions. These refinements can tighten residual sensitivity.

Next, we characterize when CI tightens local sensitivity (Eq.~\ref{eq:ls-distance}). Recall
\textbf{P2}: since RS is determined by the maximizing private relation, these
reductions matter only when they affect that relation.
\\\noindent\textbf{Rule language.}
For a relation $R_i$, we consider the following grammar of public
information\footnote{We follow the inference system of DP4SQL~\cite{dp4sql}}:
\[
P ::=
\bot
\mid R_i
\mid \select{\varphi}{R_i}
\mid \project{\boundary{i}}{R_i}
\mid \projectselect{\boundary{i}}{\varphi}{R_i}.
\]
Here, $\bot$ means no information about $R_i$ is public, while $R_i$ makes
the entire relation public. $\select{\varphi}{R_i}$ makes public only tuples
satisfying predicate $\varphi$, and $\project{\boundary{i}}{R_i}$ makes public
only the boundary attributes $\boundary{i}$ relevant to residual sensitivity.
Finally, $\projectselect{\boundary{i}}{\varphi}{R_i}$ combines both, making
public the boundary values of tuples satisfying $\varphi$.

\noindent\textbf{Neighboring edits.}
We consider three elementary edits to $R_i$: insertion $\Ins{i}$, deletion $\Del{i}$, and value change $\Chg{i}$. Under add/remove DP, $\tau\in\{\Ins{i},\Del{i}\}$; under change DP, $\tau=\Chg{i}$. We also allow combinations of these edits to capture more general neighborhood models.

\noindent\textbf{Base rule.}
We write $\Gamma,\public \vdash_I R_i:\tau\Rightarrow \ell_i$ to mean that,
on database instance $I$, under schema constraints $\Gamma$ and the
information $\public$ authorized for the current query, an allowed
neighboring edit $\tau$ to $R_i$ can change the query output by at most
$\ell_i$. When $\public=\bot$, no information about $R_i$ is public, so no refinement is
possible and we recover its ordinary local-sensitivity contribution:
\[
\frac{\public=\bot}
{\Gamma,\public \vdash_I R_i:\tau\Rightarrow\LS{i}}
\tag{\textsc{Base}}
\]
The \textsc{Base} rule therefore captures the standard DP case. The
remaining rules tighten $\LS{i}$ when the public view $\public$ restricts the
neighboring changes that are possible for $R_i$.
\\\noindent\textbf{Inference rules.}\label{sec:rules}
The box below shows the complete set of inference rules. Each rule derives
the remaining sensitivity contribution $\ell_i$ of relation $R_i$ under a
particular public view and neighboring edit; \textsc{FK} additionally
applies when $R_i$ references the excluded relation $R_E$ through a foreign
key, so referential integrity eliminates the non-occurring contribution.

\begin{mybox2}{}
\footnotesize
\begin{itemize}[nosep, leftmargin=1.1em, itemsep=7pt, label=\textbullet]
\item $\dfrac{\public=\bot}
       {\Gamma,\public\vdash_I R_i:\tau\Rightarrow\LS{i}}$
\hfill\textsc{(Base)}
\item $\dfrac{\public=\select{\varphi}{R_i}\quad\tau=\Del{i}}
       {\Gamma,\public\vdash_I R_i:\tau\Rightarrow\mufree{1}}$
\hfill\textsc{(Sel-Del)}
\item $\dfrac{\public=\project{\boundary{i}}{R_i}\quad\tau=\Chg{i}}
       {\Gamma,\public\vdash_I R_i:\tau\Rightarrow 0}$
\hfill\textsc{(Prj-Chg)}
\item $\dfrac{\public=\projectselect{\boundary{i}}{\varphi}{R_i}\quad\tau=\Chg{i}}
       {\Gamma,\public\vdash_I R_i:\tau\Rightarrow\max\{\witness{i}-\mufreemin,\mufree{1}\}}$
\hfill\textsc{(PrjSel-Chg)}
\item $\dfrac{\tau=\Ins{i}}
       {\Gamma,\public\vdash_I R_i:\tau\Rightarrow\LS{i}}$
\hfill\textsc{(Ins)}
\item $\dfrac{\public=\select{\varphi}{R_i}\quad\tau=\{\Del{i},\Chg{i}\}}
       {\Gamma,\public\vdash_I R_i:\tau\Rightarrow\max\{\witness{i}-\mufreemin,\mufree{1}\}}$
\hfill\textsc{(Sel)}
\item $\dfrac{\substack{\public=\select{\varphi}{R_i}\quad\Gamma\models\Unique(\boundary{i})\\[2pt]\tau\in\{\Ins{i},\Del{i},\Chg{i}\}}}
       {\Gamma,\public\vdash_I R_i:\tau\Rightarrow\max\{\mufree{1},\Munocci\}}$
\hfill\textsc{(Sel-Unique)}
\item $\dfrac{\public=R_i\quad\tau\in\{\Ins{i},\Del{i},\Chg{i}\}}
       {\Gamma,\public\vdash_I R_i:\tau\Rightarrow 0}$
\hfill\textsc{(Pub)}
\item $\dfrac{R_i\xrightarrow{\mathrm{FK}}R_E}
       {\Gamma\models\Munocci=0}$
\hfill\textsc{(FK)}
\end{itemize}
\end{mybox2}

\noindent \textbf{Interpreting the Rules.}
Recall that $\witness{i}=\LS{i}$ is the largest boundary multiplicity for
$R_i$. Among boundary values that remain private, let $\mufree{1}$ and
$\mufreemin$ denote the largest and smallest occupied multiplicities, and let
$\Munocci$ denote the largest multiplicity at an unoccupied boundary value.
The rules characterize which of these multiplicities remain reachable under
contextual neighborhood. \textsc{Sel-Del} excludes public tuples from deletion,
leaving $\mufree{1}$ as the largest removable contribution.
\textsc{Prj-Chg} fixes all boundary values and therefore gives zero sensitivity,
while \textsc{PrjSel-Chg} and \textsc{Sel} allow changes among the remaining
private values, yielding
$\max\{\witness{i}-\mufreemin,\mufree{1}\}$.
\textsc{Ins} retains ordinary sensitivity because a new tuple may introduce an
unrestricted boundary value; under uniqueness, \textsc{Sel-Unique} instead
limits this to $\max\{\mufree{1},\Munocci\}$.
\textsc{FK} sets $\Munocci=0$ when referential integrity prevents new
boundary values. We give a detailed explanation of these rules in
App.~\ref{app:cidp}. Using these rules, we define \emph{contextual local sensitivity} (CLS) and
derive \emph{contextual residual sensitivity} (CRS), a smooth upper bound on
CLS that is no larger than standard residual sensitivity.

\begin{theorem}[Contextual residual sensitivity]
\label{thm:crs}
For every multiway-join counting query $q$, database $D$, authorization view
$\public$, and smoothing parameter $\beta>0$,
$\CRS{q,\public}(D)$ is a $\beta$-smooth upper bound on
$\CLS{q,\public}(D)$ and satisfies
\[
    \CRS{q,\public}(D)\leq \RS{q}(D).
\]
Consequently, calibrating the release mechanism to
$\CRS{q,\public}(D)$ yields $(\epsilon,\delta)$-contextual DP under
$\public$. The construction and proof are given in
App.~\ref{app:proofs}.
\end{theorem}

\begin{corollary}[End-to-end privacy]
\label{cor:end-to-end}
Since $(\epsilon,\delta)$-DP implies $(\epsilon,\delta)$-contextual DP under
any authorization view $\public$, if the synthetic release is
$(\varepsilon_{\mathrm{syn}},\delta_{\mathrm{syn}})$-DP and the final query
release is $(\varepsilon_{\mathrm{rel}},\delta_{\mathrm{rel}})$-contextual DP
under $\public$, then \sysname{} is
$(\varepsilon_{\mathrm{syn}}+\varepsilon_{\mathrm{rel}},
\delta_{\mathrm{syn}}+\delta_{\mathrm{rel}})$-contextual DP under $\public$.
The proof is given in App.~\ref{app:proofs}.
\end{corollary}

\begin{table}[t]
\centering
\scriptsize
\setlength{\tabcolsep}{2pt}
\renewcommand{\arraystretch}{1.12}
\begin{tabular}{@{}llll@{}}
\toprule
\textbf{Mode} & \textbf{Data / Release} & \textbf{Pipeline} & \textbf{Guarantee} \\
\midrule

CI
& $\Vex$ / exact
& \makecell[l]{\textsc{Text-to-CQL}\\+ T2SQL}
& \makecell[l]{CI\\(Thm.~\ref{thm:ci-equivalence})} \\

\midrule

DP
& $D$ / noisy
& DP-aware T2SQL
& \makecell[l]{$(\epsilon,\delta)$-DP\\
(Thm.~\ref{thm:end-to-end-dp})} \\

\midrule

CI+DP: Split
& \makecell[l]{$\Vex$ exact\\+ rest noisy}
& \makecell[l]{\textsc{Text-to-CQL}\\+ DP-aware T2SQL}
& \makecell[l]{CI + $(\epsilon,\delta)$-DP\\
(Thms.~\ref{thm:end-to-end-dp}, ~\ref{thm:ci-equivalence})} \\

\midrule

CI+DP: Context.
& $D$ / noisy
& \makecell[l]{\textsc{Text-to-CQL} + \\DP-aware T2SQL + CDP}
& \makecell[l]{
$(\epsilon_{\rm syn}{+}\epsilon_{\rm rel},$\\
$\delta_{\rm syn}{+}\delta_{\rm rel})$-CDP\\
(Cor.~\ref{cor:end-to-end})
} \\

\bottomrule
\end{tabular}
\caption{End-to-end privacy modes supported by \sysname{}}.
\label{tab:privacy-modes}
\end{table}
\subsection{End-to-End Privacy}
\label{sec:e2e}

\begin{table}[t]
\centering
\caption{\small{Text-to-SQL performance on \dpdataset{}. T1/T10 denote top-1/top-10 candidates; Joint is $\mathrm{Is\mbox{-}Min}\wedge\mathrm{Acc.}$; and $\Delta$ is the Joint improvement over zero-shot T1. \textbf{Bold} marks the best result and \underline{underline} the best baseline.}}
\label{tab:text2sql-performance}

\small
\renewcommand{\arraystretch}{0.95}

\scalebox{0.7}{\begin{tabular*}{\columnwidth}{@{\extracolsep{\fill}}lrrrr@{}}
\toprule
Method & Acc. & Is-Min & Joint & $\Delta$ \\
\midrule

\multicolumn{5}{l}{\textit{XiYanSQL-32B}} \\
Zero-shot
& 77.7/70.3
& 50.5/55.5
& 49.1/52.8
& 0.0/$+3.7$ \\

Sens. prompt
& 64.0/68.2
& 38.7/50.5
& 37.2/47.9
& $-11.9$/$-1.2$ \\

SFT
& 78.9/\underline{81.5}
& 50.1/\underline{62.7}
& 47.5/\underline{59.9}
& $-1.6$/$+10.8$ \\

DPO
& 67.3/65.1
& 37.9/47.0
& 36.9/44.1
& $-12.2$/$-5.0$ \\

DPO on SFT
& 73.2/73.9
& 45.7/49.2
& 45.2/48.5
& $-3.9$/$-0.6$ \\

AlphaSQL
& 80.4 & 39.8 & 35.6 & $-13.5$ \\

\textbf{\sysname{}}
& \textbf{81.8}
& \textbf{76.6}
& \textbf{75.3}
& \textbf{+26.2} \\

\midrule
\multicolumn{5}{l}{\textit{GPT-5.6-Luna}} \\
Zero-shot
& 79.7/79.7
& 32.1/39.8
& 28.7/35.6
& 0.0/$+6.9$ \\

Sens. prompt
& 76.0/76.6
& 30.1/39.5
& 30.1/38.5
& $+1.4$/$+9.8$ \\

AlphaSQL
& \underline{\textbf{84.2}}
& \underline{47.9}
& \underline{43.9}
& $+15.2$ \\

\textbf{\sysname{}}
& 83.0
& \textbf{77.2}
& \textbf{74.6}
& \textbf{+45.9} \\

\midrule
\multicolumn{5}{l}{\textit{Phi-4}} \\
Zero-shot
& 71.3/72.9
& 46.4/\underline{53.3}
& 42.1/\underline{47.9}
& 0.0/$+5.8$ \\

Sens. prompt
& 63.4/65.9
& 30.6/43.3
& 28.6/38.6
& $-13.5$/$-3.5$ \\

AlphaSQL
& \underline{77.4}
& 41.2
& 34.9
& $-7.2$ \\

\textbf{\sysname{}}
& \textbf{79.6}
& \textbf{74.1}
& \textbf{71.2}
& \textbf{+29.1} \\

\bottomrule
\end{tabular*}}
\end{table}

We have developed three components of \sysname{}: DP-aware Text-to-SQL for
query generation, CI-aware Text-to-SQL for policy compilation, and contextual
DP (CDP). Together, they provide three privacy-preserving release modes, exposed as
a choice to the application. All modes receive the same natural-language
inputs: the text corpus, privacy policy $\Pol$, and query.
\\\noindent\textbf{CI Only.}
When the application needs contextual access control but permits exact
disclosure, \sysname{} instantiates the CI views for the query context and
computes the authorized view $\Vex$. The query is evaluated against $\Vex$, so
only authorized data is released exactly, satisfying CI by
Thm.~\ref{thm:ci-equivalence}.
\\\noindent\textbf{DP Only.}
When the full population should be protected uniformly, \sysname{} uses
DP-aware Text-to-SQL to select a correct, low-sensitivity
formulation $q$, and releases a noisy answer
satisfying standard DP by Thm.~\ref{thm:end-to-end-dp}.
\\\noindent\textbf{CI+DP.}
When a query spans both CI-authorized and private information, \sysname{} first
derives the authorized view $\Vex$ and supports two release semantics.
(i) \emph{Split release} treats the portions separately: CI-authorized
information is released exactly, while the remainder is released under
standard DP. This applies when the application can meaningfully return the two
parts separately. (ii) \emph{Contextual DP} instead releases a single
statistic over the full query. It holds the authorized information $\public$
fixed in the neighboring relation and calibrates noise using contextual
residual sensitivity, potentially requiring less noise than standard DP. By
Cor.~\ref{cor:end-to-end}, if synthetic-data generation is
$(\epsilon_{\mathrm{syn}},\delta_{\mathrm{syn}})$-DP and the final release is
$(\epsilon_{\mathrm{rel}},\delta_{\mathrm{rel}})$-contextual DP, the complete
pipeline is $(\epsilon_{\mathrm{syn}}+\epsilon_{\mathrm{rel}},
\delta_{\mathrm{syn}}+\delta_{\mathrm{rel}})$-contextual DP. Tab.~\ref{tab:privacy-modes} summarizes the guarantees.
\\\noindent\textbf{When to use each mode.}
The choice depends on the intended release. In our hospital example,
\emph{CI only} mode lets a resident receive authorized diagnoses for their care-team
patients exactly. \emph{DP only} mode instead supports aggregates e.g.,``How many
patients have HIV?'', protecting every patient with noise.
\emph{CI+DP with split release} applies when the portions can be separated:
a doctor may receive exact statistics for their own patients and a
DP-protected aggregate for the rest. \emph{Contextual DP} applies when a single aggregate is needed: a researcher authorized to access their study patients can receive one hospital-wide noisy count, with noise calibrated only to records they are not authorized to see.

\begin{table}[t]
\centering
\small
\footnotesize
\caption{\textsc{Text-to-CQL} performance across datasets.}
\scalebox{0.7}{\begin{tabular}{llccc}
\toprule
Dataset & Dimension & Precision & Recall & F1 \\
\midrule
\multirow{4}{*}{\dataset{1}}
    & Role     & 62.1\%  & 61.1\%  & 61.5\% \\
    & Data     & 83.3\%  & 97.3\%  & 86.5\% \\
    & Conditions & 95.6\%  & 76.6\%  & 76.3\% \\
    & \textbf{All} & \textbf{68.2\%} & \textbf{70.3\%} & \textbf{68.0\%} \\
\midrule
\multirow{4}{*}{\dataset{2}}
    & Role     & 96.1\%  & 68.0\%  & 71.7\% \\
    & Data     & 77.4\%  & 71.8\%  & 70.0\% \\
    & Conditions & 83.7\%  & 67.2\%  & 68.2\% \\
    & \textbf{All} & \textbf{91.1\%} & \textbf{68.9\%} & \textbf{71.2\%} \\
\midrule
\multirow{4}{*}{\dataset{3}}
    & Role     & 100.0\% & 81.7\%  & 88.6\% \\
    & Data     & 95.1\%  & 100.0\% & 97.3\% \\
    & Conditions & 97.9\%  & 82.1\%  & 85.9\% \\
    & \textbf{All} & \textbf{98.6\%} & \textbf{85.4\%} & \textbf{89.7\%} \\
\bottomrule
\end{tabular}}
\label{tab:policy-grounding}
\end{table}

\section{Evaluation}
\label{sec:evaluation}
Our evaluation is organized around the following questions:
\begin{packeditemize}
\item \textbf{RQ1.} Can \sysname{} find correct SQL with
lower sensitivity than
existing methods?
\item \textbf{RQ2.} Can \sysname{} perform \textsc{Text-to-CQL}, faithfully
translating CI policies into accurate SQL views?
\item \textbf{RQ3.} Does contextual DP reduce noise on contextually
authorized information compared to standard DP? 
\item \textbf{RQ4.} When composed end to end, does \sysname{} adhere more faithfully to the policy than baselines?
\end{packeditemize}
\noindent\textbf{Benchmarks.}
No existing benchmark jointly provides natural language (i) text, (ii) queries, and (iii) policies together with ground truth databases, SQL queries annotated for sensitivity, query contexts, and policy grounded SQL views. We therefore construct two benchmarks by combining existing datasets with new annotations. \dpdataset{} contains natural language questions paired with multiple correct SQL formulations, their Residual Sensitivity values, and the minimum sensitivity formulation $q^\star$, with 3,673 training and 886 test questions. Our CI benchmark contains 450 query context instances across healthcare (\dataset{1}), consumer (\dataset{2}), and children's privacy (\dataset{3}), with 150 instances each. Each instance includes the natural language data, privacy policy, query, runtime context, and ground truth authorized output. We use \dpdataset{} for RQ1 and \dataset{1}--\dataset{3} for RQ2--RQ4. See App.~\ref{app:benchmark} for details.
\subsection{RQ1: DP-Aware Text2SQL}
\label{sec:eval-rq1}

\noindent\textbf{Setup.}
We use $(\epsilon_{\mathrm{syn}},\delta_{\mathrm{syn}})=(1.6,10^{-6})$ for
synthetic data and $(\epsilon,\delta)=(1.0,10^{-6})$ per query, with 10
rollouts and expansion size 5. We evaluate \textsc{XiYanSQL-QwenCoder-32B}, \textsc{Phi-4}, and
\textsc{GPT-5.6-Luna} against \textbf{six} baselines: \textit{Zero-shot}; \textit{Sensitivity
Prompt}, which prompts for low sensitivity; \textit{Fine-tuning} on
XiYanSQL-32B, including SFT~\cite{sft}, DPO~\cite{dpo}, and DPO on SFT; \textit{Top-$k$ Selection}
($k=10$), which selects the lowest-sensitivity candidate; and
\textit{AlphaSQL}~\cite{li2025alpha}, which uses MCTS without our
sensitivity-guided search. We report execution accuracy (\textit{Acc.}),
minimum-sensitivity rate (\textit{Is-Min}), and their intersection
(\textit{Joint}) in Tab.~\ref{tab:text2sql-performance}. DP noise reduction, inference costs, and ablations are
reported in App.~\ref{app:eval}.

\noindent\textbf{Results.}
Table~\ref{tab:text2sql-performance} shows that \sysname{} substantially improves the joint objective across all three model families, achieving 75.3\%, 74.6\%, and 71.2\% Joint accuracy with XiYanSQL-32B, GPT-5.6-Luna, and Phi-4, respectively. This improves over zero-shot by 26.2, 45.9, and 29.1 points, and over the strongest baseline by 15.4, 30.7, and 23.3 points. In contrast, AlphaSQL achieves high accuracy but achieves only 39.8--47.9\% Is-Min. Thus, search alone is insufficient; explicitly steering it toward low-sensitivity formulations is necessary to jointly preserve correctness and DP utility.

\subsection{RQ2: CI-Aware Text2SQL}
\label{sec:eval-rq2}

\noindent\textbf{Setup.}
We evaluate whether \sysname{} can compile natural-language CI norms into SQL views, i.e., \textsc{Text-to-CQL}, by comparing the compiled views against manually annotated ground truth. We report precision, recall, and F1 scores.

\noindent\textbf{Results.}
Tab.~\ref{tab:policy-grounding} shows high precision across all mapping types (62.1--100\%), with lower recall indicating that \sysname{} more often misses valid mappings than introduces incorrect ones. Performance is strongest on \dataset{3}, reaching F1 scores of 88.6\%, 97.3\%, and 85.9\% for role, data, and condition mappings, respectively.

\subsection{RQ3: Contextual DP}
\label{sec:eval-rq3}

\noindent\textbf{Setup.}
In Sec.~\ref{sec:ci-dp}, we characterize the conditions under which
contextual DP reduces noise. Since these patterns do not naturally arise under
our benchmark policies, we validate the derived cases by progressively
authorizing records in each dataset's primary subject relation and comparing
the average required noise across all queries under contextual and standard DP
with the same privacy budget, $\epsilon=1$.

\noindent\textbf{Results.}
Fig.~\ref{fig:rq3} validates our characterization: contextual DP never requires
more noise than standard DP. At 90\% authorization, mean noise falls to
$0.44\times$ for \dataset{2}, $0.60\times$ for \dataset{1}, and $0.89\times$
for \dataset{3}, corresponding to reductions of 56\%, 40\%, and 11\%,
respectively.

\begin{figure}
\centering
\includegraphics[width=0.6\linewidth]{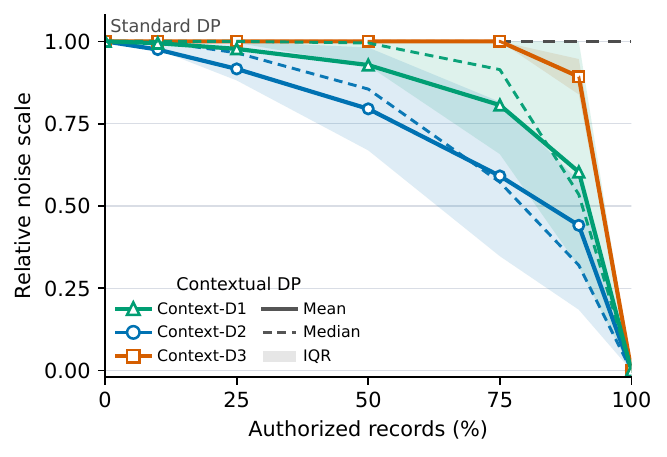}
\caption{\footnotesize Relative noise required by contextual DP as we authorize an increasing fraction of subject records.}
\label{fig:rq3}
\vspace{-3mm}
\end{figure}

\subsection{RQ4: End-to-End Evaluation}
\label{sec:eval-rq4}

\noindent\textbf{Setup.}
We evaluate \sysname{} end to end, from natural-language data and policy to the
final answer released to the user. We compare the resulting
\emph{information leakage} under four approaches using GPT-5.6-Luna:
\textit{Text-only}, which answers directly from the text without the policy;
\textit{Text+Policy}, which gives the LLM the policy and context and asks it to
enforce them; \textit{AirGapAgent}~\cite{bagdasarian2024airgapagent}, which first uses a trusted LLM to filter the
text to information it considers authorized; and \sysname{}. We
report \emph{leakage rate}, the fraction of responses that reveal information
prohibited by the ground-truth policy for that query.

\begin{table}[t]
\centering
\small
\setlength{\tabcolsep}{4pt}
\renewcommand{\arraystretch}{0.95}
\scalebox{0.7}{\begin{tabular*}{\columnwidth}{@{\extracolsep{\fill}}lccc@{}}
\toprule
System & \dataset{1} & \dataset{2} & \dataset{3} \\
\midrule
Text-only                  & 88.7\%  & 85.2\% & 100.0\% \\
Text+Policy                & 18.0\%  & 52.3\% & 44.3\% \\
AirGapAgent                & 6.7\%   & 38.9\% & 7.4\% \\
\textbf{\sysname{}}        & \textbf{0.0\%} & \textbf{2.1\%} & \textbf{1.3\%} \\
\sysname{}$^\dagger$ & 0.0\%   & 0.0\%  & 0.0\% \\
\bottomrule
\end{tabular*}}
\caption{\footnotesize End-to-end prohibited information leakage using GPT-5.6-Luna. \sysname{}$^\dagger$ uses the ground-truth policy views. Leakage in \sysname{} can additionally be protected under DP rather
than released exactly.}
\label{tab:leakage}
\end{table}

\noindent\textbf{Results.}
Tab.~\ref{tab:leakage} shows a clear reduction in leakage as enforcement becomes
more explicit. Giving the policy directly to the LLM (Text+Policy) still leaks
prohibited information in 18.0--52.3\% of queries, while AirGapAgent reduces
this to 6.7--38.9\%. In contrast, \sysname{} reduces leakage to 0.0\%, 2.1\%,
and 1.3\% on \dataset{1}, \dataset{2}, and \dataset{3}, respectively. The
leakage in \sysname{} comes entirely from policy-view compilation errors: with
ground-truth views, deterministic enforcement achieves 0\% leakage on all
datasets. Moreover, the leakage in \sysname{} can still be
protected under DP rather than released exactly.
\section{Related Work}\label{sec:relatedwork}
\noindent \textbf{CI for text.}
Prior work extracts or annotates CI norms from natural-language policies~\cite{shvartzshanider-etal-2018-recipe,shvartzshnaider2019going,shvartzshanider2023beyond,chanenson2025automating}, uses CI as LLM context for leakage detection~\cite{mireshghallah2024can,fan2024goldcoin,xiao2024contextual,bagdasarian2024airgapagent,ghalebikesabi2025privacy,li2025privacy,li-etal-2025-privaci,lan2025contextual,wang2025privacy,huang2026need}. \sysname{} instead compiles CI policies into SQL views.

\noindent \textbf{DP for text.}
DP for text~\cite{klymenko-etal-2022-differential,hu2024differentially} spans local sanitization of text~\cite{feyisetan2020privacy,xu-etal-2020-differentially,feyisetan-kasiviswanathan-2021-private,meehan2022sentence,mattern2022limits,awon2025clusant,zhang2025dyntext,chowdhury2025pr}, private language-model training~\cite{mcmahan2018learning,li2022large,yu2022differentially,anil2022large}, and private synthetic-text generation~\cite{mattern2022language,yue2023synthetic}, often requiring substantial noise that degrades utility~\cite{yue2021differentialprivacytextanalytics};.  \sysname{} instead grounds text into a relational representation and privatizes query answers. Prior SQL-DP work optimizes execution plans, query rewrites, sensitivity bounds, or DP mechanisms~\cite{johnson2018towards,johnson2020chorus,kotsogiannis2019privatesql,dong2021residual,dong2022nearly,dong2022r2t}. \sysname{} instead applies DP to SQL queries, building on prior work on private query answering~\cite{johnson2018towards,kotsogiannis2019privatesql,dong2021residual}, but uniquely optimizes privacy by searching over query formulations from natural language.

\noindent \textbf{Policy-aware DP.}
Customized DP adapts protection using explicit secrets or constraints~\cite{kifer2014pufferfish,blowfish,haney2016policy}, sensitive records or text spans~\cite{shi2022selective}, known relational privacy policies~\cite{dp4sql}. These approaches take protection semantics as given; \sysname{} instead derives querier- and context-specific authorization from natural-language CI policies and uses it to define the DP neighborhood. Prior work has connected CI and DP conceptually~\cite{rachel-cummings-dp}; \sysname{} makes this connection operational.

\section{Conclusion}\label{sec:conclusion}
We presented \sysname{}, a framework for formally enforcing CI and DP over unstructured text through a shared relational representation. \sysname{} supports DP aware query generation, executable CI policies, and their composition through \emph{contextual differential privacy}.

\section*{Ethical Considerations}

This paper studies privacy-preserving analytics over sensitive unstructured text using contextual integrity (CI), differential privacy (DP), or their composition. Because the target setting includes domains such as healthcare and consumer data, failures may expose sensitive information. We therefore consider the ethical implications of both the system and its deployment.

\noindent\textbf{Stakeholders.}
Stakeholders include individuals represented in the underlying text, data custodians, analysts issuing queries, administrators defining privacy policies, and the research community.

\noindent\textbf{Principles.}
We consider the Menlo Report principles of Beneficence, Respect for Persons, Justice, and Respect for Law and Public Interest. Our goal is to enable useful analysis while limiting inappropriate information flows and inference about individuals.

\noindent\textbf{Potential harms.}
The primary risk is unintended disclosure of sensitive information. While \sysname{} enforces CI and DP through explicit mechanisms rather than LLM judgments, errors in grounding natural-language data, queries, or policies may still affect correctness or privacy. 

The system also has potential for dual use. Privacy-preserving access to sensitive organizational text could be deployed under overly permissive policies or where affected individuals have limited control over how their data is analyzed. Formal enforcement guarantees adherence to the supplied policy, not that the policy itself is ethically appropriate.

\noindent\textbf{Mitigations.}
We mitigate these risks by making privacy enforcement explicit and auditable. CI policies are compiled into database views that can be inspected and validated before deployment, while runtime authorization is enforced deterministically. DP releases provide formally bounded disclosure, and contextual DP applies only after determining the information authorized for the requesting context.

Our experiments do not recruit or interact with human participants. We evaluate on existing Text-to-SQL benchmarks and public privacy policies, including HIPAA, CCPA, and COPPA, with additional annotations for evaluating privacy enforcement.

\noindent\textbf{Decision.}
We believe it is ethical to proceed and publish because the work aims to reduce privacy risks in natural-language analytics while making its assumptions and failure modes explicit. Publication enables scrutiny of these guarantees and provides a more auditable alternative to approaches that delegate privacy decisions directly to LLMs.

\bibliographystyle{plainurl}
\bibliography{custom}

@inproceedings{mireshghallah2024can,
  title={Can LLMs Keep a Secret? Testing Privacy Implications of Language Models via Contextual Integrity Theory},
  author={Mireshghallah, Niloofar and Kim, Hyunwoo and Zhou, Xuhui and Tsvetkov, Yulia and Sap, Maarten and Shokri, Reza and Choi, Yejin},
  booktitle={International Conference on Learning Representations},
  year={2024}
}

@inproceedings{rachel-cummings-dp,
  title={Integrating differential privacy and contextual integrity},
  author={Benthall, Sebastian and Cummings, Rachel},
  booktitle={Proceedings of the 2024 Symposium on Computer Science and Law},
  pages={9--15},
  year={2024}
}

@inproceedings{barth2006privacy,
  title={Privacy and Contextual Integrity: Framework and Applications},
  author={Barth, Adam and Datta, Anupam and Mitchell, John C and Nissenbaum, Helen},
  booktitle={2006 IEEE Symposium on Security and Privacy},
  pages={184--198},
  year={2006},
  publisher={IEEE},
  doi={10.1109/SP.2006.32}
}

@article{dwork2014algorithmic,
  title={The Algorithmic Foundations of Differential Privacy},
  author={Dwork, Cynthia and Roth, Aaron},
  journal={Foundations and Trends in Theoretical Computer Science},
  volume={9},
  number={3--4},
  pages={211--407},
  year={2014},
  doi={10.1561/0400000042}
}

@inproceedings{dong2021residual,
  title={Residual Sensitivity for Differentially Private Multi-Way Joins},
  author={Dong, Wei and Yi, Ke},
  booktitle={Proceedings of the 2021 International Conference on Management of Data},
  pages={432--444},
  year={2021},
  publisher={ACM},
  doi={10.1145/3448016.3452813}
}

@article{cai2025privpetal,
  title={{PrivPetal}: Relational Data Synthesis via Permutation Relations},
  author={Cai, Kuntai and Xiao, Xiaokui and Yang, Yin},
  journal={Proceedings of the ACM on Management of Data (SIGMOD)},
  volume={3},
  number={3},
  articleno={204},
  year={2025},
  publisher={ACM},
  doi={10.1145/3725341}
}

@inproceedings{kocsis2006bandit,
  title={Bandit Based Monte-Carlo Planning},
  author={Kocsis, Levente and Szepesv{\'a}ri, Csaba},
  booktitle={Machine Learning: ECML 2006},
  series={Lecture Notes in Computer Science},
  volume={4212},
  pages={282--293},
  year={2006},
  publisher={Springer},
  doi={10.1007/11871842_29}
}

@inproceedings{bagdasarian2024airgapagent,
  title={Airgapagent: Protecting privacy-conscious conversational agents},
  author={Bagdasarian, Eugene and Yi, Ren and Ghalebikesabi, Sahra and Kairouz, Peter and Gruteser, Marco and Oh, Sewoong and Balle, Borja and Ramage, Daniel},
  booktitle={Proceedings of the 2024 on ACM SIGSAC Conference on Computer and Communications Security},
  pages={3868--3882},
  year={2024}
}

@article{lan2025contextual,
  title={Contextual integrity in LLMs via reasoning and reinforcement learning},
  author={Lan, Guangchen and Inan, Huseyin A and Abdelnabi, Sahar and Kulkarni, Janardhan and Wutschitz, Lukas and Shokri, Reza and Brinton, Christopher G and Sim, Robert},
  journal={arXiv preprint arXiv:2506.04245},
  year={2025}
}

@inproceedings{tsai2025contextual,
  title={Contextual agent security: A policy for every purpose},
  author={Tsai, Lillian and Bagdasarian, Eugene},
  booktitle={Proceedings of the 2025 Workshop on Hot Topics in Operating Systems},
  pages={8--17},
  year={2025}
}

@article{abdelnabi1822firewalls,
  title={Firewalls to secure dynamic llm agentic networks, 2025},
  author={Abdelnabi, Sahar and Gomaa, Amr and Bagdasarian, Eugene and Kristensson, Per Ola and Shokri, Reza},
  journal={URL https://arxiv. org/abs/2502},
  year={1822}
}

@article{ghalebikesabi2025privacy,
  title={Privacy awareness for information-sharing assistants: A case-study on form-filling with contextual integrity},
  author={Ghalebikesabi, Sahra and Bagdasarian, Eugene and Yi, Ren and Yona, Itay and Shumailov, Ilia and Pappu, Aneesh and Shi, Chongyang and Weidinger, Laura and Stanforth, Robert and Berrada, Leonard and others},
  journal={Transactions on Machine Learning Research},
  year={2025}
}

@article{siyan2024papillon,
  title={Papillon: Privacy preservation from internet-based and local language model ensembles},
  author={Siyan, Li and Raghuram, Vethavikashini Chithrra and Khattab, Omar and Hirschberg, Julia and Yu, Zhou},
  journal={arXiv preprint arXiv:2410.17127},
  year={2024}
}

@article{yi2025privacy,
  title={Privacy Reasoning in Ambiguous Contexts},
  author={Yi, Ren and Suciu, Octavian and Gascon, Adria and Meiklejohn, Sarah and Bagdasarian, Eugene and Gruteser, Marco},
  journal={arXiv preprint arXiv:2506.12241},
  year={2025}
}

@inproceedings{ngong2025protecting,
  title={Protecting users from themselves: Safeguarding contextual privacy in interactions with conversational agents},
  author={Ngong, Ivoline C and Kadhe, Swanand Ravindra and Wang, Hao and Murugesan, Keerthiram and Weisz, Justin D and Dhurandhar, Amit and Ramamurthy, Karthikeyan Natesan},
  booktitle={Findings of the Association for Computational Linguistics: ACL 2025},
  pages={26196--26220},
  year={2025}
}

@inproceedings{klymenko-etal-2022-differential,
    title = "Differential Privacy in Natural Language Processing: The Story So Far",
    author = "Klymenko, Oleksandra  and
      Meisenbacher, Stephen  and
      Matthes, Florian",
    editor = "Feyisetan, Oluwaseyi  and
      Ghanavati, Sepideh  and
      Thaine, Patricia  and
      Habernal, Ivan  and
      Mireshghallah, Fatemehsadat",
    booktitle = "Proceedings of the Fourth Workshop on Privacy in Natural Language Processing",
    month = jul,
    year = "2022",
    address = "Seattle, United States",
    publisher = "Association for Computational Linguistics",
    url = "https://aclanthology.org/2022.privatenlp-1.1/",
    doi = "10.18653/v1/2022.privatenlp-1.1",
    pages = "1--11"
}

@inproceedings{hu-etal-2024-differentially,
    title = "Differentially Private Natural Language Models: Recent Advances and Future Directions",
    author = "Hu, Lijie  and
      Habernal, Ivan  and
      Shen, Lei  and
      Wang, Di",
    editor = "Graham, Yvette  and
      Purver, Matthew",
    booktitle = "Findings of the Association for Computational Linguistics: EACL 2024",
    month = mar,
    year = "2024",
    address = "St. Julian{'}s, Malta",
    publisher = "Association for Computational Linguistics",
    url = "https://aclanthology.org/2024.findings-eacl.33/",
    doi = "10.18653/v1/2024.findings-eacl.33",
    pages = "478--499"
}

@inproceedings{igamberdiev2023dp,
  title={DP-BART for privatized text rewriting under local differential privacy},
  author={Igamberdiev, Timour and Habernal, Ivan},
  booktitle={Findings of the Association for Computational Linguistics: ACL 2023},
  pages={13914--13934},
  year={2023}
}

@inproceedings{squid,
    title = "{SQU}i{D}: Synthesizing Relational Databases from Unstructured Text",
    author = "Sadia, Mushtari  and
      Yang, Zhenning  and
      Xiao, Yunming  and
      Chen, Ang  and
      Roy Chowdhury, Amrita",
    editor = "Christodoulopoulos, Christos  and
      Chakraborty, Tanmoy  and
      Rose, Carolyn  and
      Peng, Violet",
    booktitle = "Proceedings of the 2025 Conference on Empirical Methods in Natural Language Processing",
    month = nov,
    year = "2025",
    address = "Suzhou, China",
    publisher = "Association for Computational Linguistics",
    url = "https://aclanthology.org/2025.emnlp-main.1629/",
    doi = "10.18653/v1/2025.emnlp-main.1629",
    pages = "31987--32012",
    ISBN = "979-8-89176-332-6"
}

@inproceedings{johnson2020chorus,
  title={Chorus: a programming framework for building scalable differential privacy mechanisms},
  author={Johnson, Noah and Near, Joseph P and Hellerstein, Joseph M and Song, Dawn},
  booktitle={2020 IEEE European Symposium on Security and Privacy (EuroS\&P)},
  pages={535--551},
  year={2020},
  organization={IEEE}
}

@inproceedings{dong2022r2t,
  title={R2t: Instance-optimal truncation for differentially private query evaluation with foreign keys},
  author={Dong, Wei and Fang, Juanru and Yi, Ke and Tao, Yuchao and Machanavajjhala, Ashwin},
  booktitle={Proceedings of the 2022 International Conference on Management of Data},
  pages={759--772},
  year={2022}
}

@article{kotsogiannis2019privatesql,
  title={Privatesql: a differentially private sql query engine},
  author={Kotsogiannis, Ios and Tao, Yuchao and He, Xi and Fanaeepour, Maryam and Machanavajjhala, Ashwin and Hay, Michael and Miklau, Gerome},
  journal={Proceedings of the VLDB Endowment},
  volume={12},
  number={11},
  pages={1371--1384},
  year={2019},
  publisher={VLDB Endowment}
}

@article{bater2018shrinkwrap,
  title={Shrinkwrap: efficient sql query processing in differentially private data federations},
  author={Bater, Johes and He, Xi and Ehrich, William and Machanavajjhala, Ashwin and Rogers, Jennie},
  journal={Proceedings of the VLDB Endowment},
  volume={12},
  number={3},
  year={2018}
}

@article{li2025alpha,
  title={{Alpha-SQL}: Zero-Shot Text-to-SQL Using Monte Carlo Tree Search},
  author={Li, Boyan and Zhang, Jiayi and Fan, Ju and Xu, Yanwei and Chen, Chong and Tang, Nan and Luo, Yuyu},
  journal={arXiv preprint arXiv:2502.17248},
  year={2025}
}

@misc{yue2021differentialprivacytextanalytics,
      title={Differential Privacy for Text Analytics via Natural Text Sanitization}, 
      author={Xiang Yue and Minxin Du and Tianhao Wang and Yaliang Li and Huan Sun and Sherman S. M. Chow},
      year={2021},
      eprint={2106.01221},
      archivePrefix={arXiv},
      primaryClass={cs.CL},
      url={https://arxiv.org/abs/2106.01221}, 
}

@misc{meisenbacher2025leveragingsemantictriplesprivate,
      title={Leveraging Semantic Triples for Private Document Generation with Local Differential Privacy Guarantees}, 
      author={Stephen Meisenbacher and Maulik Chevli and Florian Matthes},
      year={2025},
      eprint={2508.20736},
      archivePrefix={arXiv},
      primaryClass={cs.CL},
      url={https://arxiv.org/abs/2508.20736}, 
}

@inproceedings{zhan2024injecagent,
  title={Injecagent: Benchmarking indirect prompt injections in tool-integrated large language model agents},
  author={Zhan, Qiusi and Liang, Zhixiang and Ying, Zifan and Kang, Daniel},
  booktitle={Findings of the Association for Computational Linguistics: ACL 2024},
  pages={10471--10506},
  year={2024}
}

@inproceedings{jiao-etal-2024-text2db,
    title = "{T}ext2{DB}: Integration-Aware Information Extraction with Large Language Model Agents",
    author = "Jiao, Yizhu  and
      Li, Sha  and
      Zhou, Sizhe  and
      Ji, Heng  and
      Han, Jiawei",
    editor = "Ku, Lun-Wei  and
      Martins, Andre  and
      Srikumar, Vivek",
    booktitle = "Findings of the Association for Computational Linguistics: ACL 2024",
    month = aug,
    year = "2024",
    address = "Bangkok, Thailand",
    publisher = "Association for Computational Linguistics",
    url = "https://aclanthology.org/2024.findings-acl.12/",
    doi = "10.18653/v1/2024.findings-acl.12",
    pages = "185--205"
}

@article{liang2026beyond,
  title={Beyond Tables: Doc2DB-Bench for Relationally Faithful Document-to-Database Construction},
  author={Liang, Zhuowen and Zhang, Zhengxuan and Wang, Jiayang and Chen, Jiazhuo and Tang, Nan},
  journal={arXiv preprint arXiv:2608.08459},
  year={2026}
}

@article{lin2026structure,
  title={Structure then Query: Enabling Precise Analytical Queries over Unstructured Documents},
  author={Lin, Teng and Luo, Yuyu and Tang, Nan},
  journal={arXiv preprint arXiv:2608.13384},
  year={2026}
}

@inproceedings{deng-etal-2024-text,
    title = "Text-Tuple-Table: Towards Information Integration in Text-to-Table Generation via Global Tuple Extraction",
    author = "Deng, Zheye  and
      Chan, Chunkit  and
      Wang, Weiqi  and
      Sun, Yuxi  and
      Fan, Wei  and
      Zheng, Tianshi  and
      Yim, Yauwai  and
      Song, Yangqiu",
    editor = "Al-Onaizan, Yaser  and
      Bansal, Mohit  and
      Chen, Yun-Nung",
    booktitle = "Proceedings of the 2024 Conference on Empirical Methods in Natural Language Processing",
    month = nov,
    year = "2024",
    address = "Miami, Florida, USA",
    publisher = "Association for Computational Linguistics",
    url = "https://aclanthology.org/2024.emnlp-main.523/",
    doi = "10.18653/v1/2024.emnlp-main.523",
    pages = "9300--9322"
}

@misc{arora2025languagemodelsenablesimple,
      title={Language Models Enable Simple Systems for Generating Structured Views of Heterogeneous Data Lakes}, 
      author={Simran Arora and Brandon Yang and Sabri Eyuboglu and Avanika Narayan and Andrew Hojel and Immanuel Trummer and Christopher Ré},
      year={2025},
      eprint={2304.09433},
      archivePrefix={arXiv},
      primaryClass={cs.CL},
      url={https://arxiv.org/abs/2304.09433}, 
}

@inproceedings{fan2024goldcoin,
  title={Goldcoin: Grounding large language models in privacy laws via contextual integrity theory},
  author={Fan, Wei and Li, Haoran and Deng, Zheye and Wang, Weiqi and Song, Yangqiu},
  booktitle={Proceedings of the 2024 Conference on Empirical Methods in Natural Language Processing},
  pages={3321--3343},
  year={2024}
}

@inproceedings{li2025privacy,
  title={Privacy checklist: Privacy violation detection grounding on contextual integrity theory},
  author={Li, Haoran and Fan, Wei and Chen, Yulin and Jiayang, Cheng and Chu, Tianshu and Zhou, Xuebing and Hu, Peizhao and Song, Yangqiu},
  booktitle={Proceedings of the 2025 Conference of the Nations of the Americas Chapter of the Association for Computational Linguistics: Human Language Technologies (Volume 1: Long Papers)},
  pages={1748--1766},
  year={2025}
}

@article{unstructured-data,
    note = {\url{https://mitsloan.mit.edu/ideas-made-to-matter/tapping-power-unstructured-data}},
    author = "Harbert, Tam",

}

@misc{wikipedia_hipaa,
  author       = {{Wikipedia contributors}},
  title        = {Health Insurance Portability and Accountability Act},
  howpublished = {Wikipedia, The Free Encyclopedia},
  year         = {2026},
  url          = {https://en.wikipedia.org/wiki/Health_Insurance_Portability_and_Accountability_Act},
  note         = {Accessed: 2026-08-24}
}

@misc{wikipedia_gdpr,
  author       = {{Wikipedia contributors}},
  title        = {General Data Protection Regulation},
  howpublished = {Wikipedia, The Free Encyclopedia},
  year         = {2026},
  url          = {https://en.wikipedia.org/wiki/General_Data_Protection_Regulation},
  note         = {Accessed: 2026-08-24}
}

@inproceedings{dwork2006differential,
  title={Differential privacy},
  author={Dwork, Cynthia},
  booktitle={International colloquium on automata, languages, and programming},
  pages={1--12},
  year={2006},
  organization={Springer}
}

@inproceedings{10.1145/3173574.3173842,
author = {Wijesekera, Primal and Reardon, Joel and Reyes, Irwin and Tsai, Lynn and Chen, Jung-Wei and Good, Nathan and Wagner, David and Beznosov, Konstantin and Egelman, Serge},
title = {Contextualizing Privacy Decisions for Better Prediction (and Protection)},
year = {2018},
isbn = {9781450356206},
publisher = {Association for Computing Machinery},
address = {New York, NY, USA},
url = {https://doi.org/10.1145/3173574.3173842},
doi = {10.1145/3173574.3173842},
booktitle = {Proceedings of the 2018 CHI Conference on Human Factors in Computing Systems},
pages = {1–13},
numpages = {13},
location = {Montreal QC, Canada},
series = {CHI '18}
}

@inproceedings{meisenbacher-etal-2024-comparative,
    title = "A Comparative Analysis of Word-Level Metric Differential Privacy: Benchmarking the Privacy-Utility Trade-off",
    author = "Meisenbacher, Stephen  and
      Nandakumar, Nihildev  and
      Klymenko, Alexandra  and
      Matthes, Florian",
    editor = "Calzolari, Nicoletta  and
      Kan, Min-Yen  and
      Hoste, Veronique  and
      Lenci, Alessandro  and
      Sakti, Sakriani  and
      Xue, Nianwen",
    booktitle = "Proceedings of the 2024 Joint International Conference on Computational Linguistics, Language Resources and Evaluation (LREC-COLING 2024)",
    month = may,
    year = "2024",
    address = "Torino, Italia",
    publisher = "ELRA and ICCL",
    url = "https://aclanthology.org/2024.lrec-main.16/",
    pages = "174--185"
}

@inproceedings{feyisetan2020privacy,
  title={Privacy-and utility-preserving textual analysis via calibrated multivariate perturbations},
  author={Feyisetan, Oluwaseyi and Balle, Borja and Drake, Thomas and Diethe, Tom},
  booktitle={Proceedings of the 13th international conference on web search and data mining},
  pages={178--186},
  year={2020}
}

@inproceedings{xu-etal-2020-differentially,
    title = "A Differentially Private Text Perturbation Method Using Regularized Mahalanobis Metric",
    author = "Xu, Zekun  and
      Aggarwal, Abhinav  and
      Feyisetan, Oluwaseyi  and
      Teissier, Nathanael",
    editor = "Feyisetan, Oluwaseyi  and
      Ghanavati, Sepideh  and
      Malmasi, Shervin  and
      Thaine, Patricia",
    booktitle = "Proceedings of the Second Workshop on Privacy in NLP",
    month = nov,
    year = "2020",
    address = "Online",
    publisher = "Association for Computational Linguistics",
    url = "https://aclanthology.org/2020.privatenlp-1.2/",
    doi = "10.18653/v1/2020.privatenlp-1.2",
    pages = "7--17"
}

@inproceedings{feyisetan-kasiviswanathan-2021-private,
    title = "Private Release of Text Embedding Vectors",
    author = "Feyisetan, Oluwaseyi  and
      Kasiviswanathan, Shiva",
    editor = "Pruksachatkun, Yada  and
      Ramakrishna, Anil  and
      Chang, Kai-Wei  and
      Krishna, Satyapriya  and
      Dhamala, Jwala  and
      Guha, Tanaya  and
      Ren, Xiang",
    booktitle = "Proceedings of the First Workshop on Trustworthy Natural Language Processing",
    month = jun,
    year = "2021",
    address = "Online",
    publisher = "Association for Computational Linguistics",
    url = "https://aclanthology.org/2021.trustnlp-1.3/",
    doi = "10.18653/v1/2021.trustnlp-1.3",
    pages = "15--27"
}

@misc{google_document_ai,
  author       = {{Google Cloud}},
  title        = {Document AI Overview},
  howpublished = {Google Cloud Documentation},
  year         = {2026},
  url          = {https://docs.cloud.google.com/document-ai/docs/overview},
  note         = {Accessed: 2026-08-24}
}

@misc{microsoft_content_understanding,
  author       = {{Microsoft}},
  title        = {Choose an Azure AI Targeted Language Processing Technology},
  howpublished = {Microsoft Learn},
  year         = {2026},
  url          = {https://learn.microsoft.com/en-us/azure/architecture/data-guide/ai-services/targeted-language-processing},
  note         = {Accessed: 2026-08-24}
}

@misc{microsoft_document_intelligence,
  author       = {{Microsoft}},
  title        = {What Is Azure Document Intelligence in Foundry Tools?},
  howpublished = {Microsoft Learn},
  year         = {2026},
  url          = {https://learn.microsoft.com/en-us/azure/ai-services/document-intelligence/overview?view=doc-intel-4.0.0},
  note         = {Accessed: 2026-08-24}
}

@misc{unstructured_extract,
  author       = {{Unstructured}},
  title        = {Introducing: Extract},
  howpublished = {Unstructured},
  year         = {2026},
  url          = {https://unstructured.io/blog/introducing-extract},
  note         = {Accessed: 2026-08-24}
}

@misc{talonic,
  author       = {{Talonic}},
  title        = {Talonic: Ingest Once. Query Forever.},
  howpublished = {Talonic},
  year         = {2026},
  url          = {https://talonic.com/},
  note         = {Accessed: 2026-08-24}
}

@misc{unstract_aws,
  author       = {{Unstract}},
  title        = {Unstract: Agentic Document Processing and Extraction Platform},
  howpublished = {AWS Marketplace},
  year         = {2026},
  url          = {https://aws.amazon.com/marketplace/pp/prodview-jlud5hmy2npui},
  note         = {Accessed: 2026-08-24}
}

@misc{snowflake_unstructured_data,
  author       = {{Snowflake}},
  title        = {Unstructured Data Analytics and AI Solutions},
  howpublished = {Snowflake},
  year         = {2026},
  url          = {https://www.snowflake.com/en/product/use-cases/unstructured-data-analytics/},
  note         = {Accessed: 2026-08-24}
}

@article{zhu2024large,
  title={Large language model enhanced text-to-sql generation: A survey},
  author={Zhu, Xiaohu and Li, Qian and Cui, Lizhen and Liu, Yongkang},
  journal={arXiv preprint arXiv:2410.06011},
  year={2024}
}

@article{hong2025next,
  title={Next-generation database interfaces: A survey of llm-based text-to-sql},
  author={Hong, Zijin and Yuan, Zheng and Zhang, Qinggang and Chen, Hao and Dong, Junnan and Huang, Feiran and Huang, Xiao},
  journal={IEEE Transactions on Knowledge and Data Engineering},
  year={2025},
  publisher={IEEE}
}

@article{shi2025survey,
  title={A survey on employing large language models for text-to-sql tasks},
  author={Shi, Liang and Tang, Zhengju and Zhang, Nan and Zhang, Xiaotong and Yang, Zhi},
  journal={ACM Computing Surveys},
  volume={58},
  number={2},
  pages={1--37},
  year={2025},
  publisher={ACM New York, NY}
}

@article{mohammadjafari2024natural,
  title={From natural language to sql: Review of llm-based text-to-sql systems},
  author={Mohammadjafari, Ali and Maida, Anthony S and Gottumukkala, Raju},
  journal={arXiv preprint arXiv:2410.01066},
  year={2024}
}

@misc{reuters2022ccpa,
  author       = {{Reuters}},
  title        = {Little breathing room: California privacy agency modifies proposed regulations},
  year         = {2022},
  month        = dec,
  day          = {8},
  url          = {https://www.reuters.com/legal/legalindustry/little-breathing-room-california-privacy-agency-modifies-proposed-regulations-2022-12-08/},
}

@misc{congruity360ccpa,
  author       = {{Congruity360}},
  title        = {CCPA Compliance Guide},
  url          = {https://www.congruity360.com/blog/ccpa-compliance-guide/}
}

@article{johnson2018towards,
  title={Towards practical differential privacy for SQL queries},
  author={Johnson, Noah and Near, Joseph P and Song, Dawn},
  journal={Proceedings of the VLDB Endowment},
  volume={11},
  number={5},
  pages={526--539},
  year={2018},
  publisher={VLDB Endowment}
}

@article{wilson2019differentially,
  title={Differentially private SQL with bounded user contribution},
  author={Wilson, Royce J and Zhang, Celia Yuxin and Lam, William and Desfontaines, Damien and Simmons-Marengo, Daniel and Gipson, Bryant},
  journal={arXiv preprint arXiv:1909.01917},
  year={2019}
}

@inproceedings{nissim2007smooth,
  title={Smooth sensitivity and sampling in private data analysis},
  author={Nissim, Kobbi and Raskhodnikova, Sofya and Smith, Adam},
  booktitle={Proceedings of the thirty-ninth annual ACM symposium on Theory of computing},
  pages={75--84},
  year={2007}
}

@misc{wikipedia_shapley_value,
  author       = {{Wikipedia contributors}},
  title        = {Shapley Value},
  howpublished = {\url{https://en.wikipedia.org/wiki/Shapley_value}},
  note         = {Wikipedia, The Free Encyclopedia},
  year         = {2026}
}

@inproceedings{yuan2026mctssql,
author = {Yuan, Shuozhi and Chen, Liming and Yuan, Miaomiao and Jin, Zhao},
title = {MCTS-SQL: light-weight LLMs can master the text-to-SQL through Monte Carlo tree search},
year = {2026},
isbn = {978-1-57735-906-7},
publisher = {AAAI Press},
url = {https://doi.org/10.1609/aaai.v40i41.40751},
doi = {10.1609/aaai.v40i41.40751},
booktitle = {Proceedings of the Fortieth AAAI Conference on Artificial Intelligence and Thirty-Eighth Conference on Innovative Applications of Artificial Intelligence and Sixteenth Symposium on Educational Advances in Artificial Intelligence},
articleno = {3850},
numpages = {9},
series = {AAAI'26/IAAI'26/EAAI'26}
}

@inproceedings{amb2,
  title={Evaluating Ambiguous Questions in Text2SQL},
  author={Papicchio, Simone and Cagliero, Luca and Papotti, Paolo},
  booktitle={ELLIS workshop on Representation Learning and Generative Models for Structured Data}
}

@inproceedings{amb1,
  title={PRACTIQ: A practical conversational text-to-SQL dataset with ambiguous and unanswerable queries},
  author={Dong, Mingwen and Kumar, Nischal Ashok and Hu, Yiqun and Chauhan, Anuj and Hang, Chung-Wei and Chang, Shuaichen and Pan, Lin and Lan, Wuwei and Zhu, Henghui and Jiang, Jiarong and others},
  booktitle={Proceedings of the 2025 Conference of the Nations of the Americas Chapter of the Association for Computational Linguistics: Human Language Technologies (Volume 1: Long Papers)},
  pages={255--273},
  year={2025}
}

@inproceedings{amb3,
  title={Benchmarking and improving text-to-SQL generation under ambiguity},
  author={Bhaskar, Adithya and Tomar, Tushar and Sathe, Ashutosh and Sarawagi, Sunita},
  booktitle={Proceedings of the 2023 Conference on Empirical Methods in Natural Language Processing},
  pages={7053--7074},
  year={2023}
}

@inproceedings{amb4,
  title={Ambisql: Interactive ambiguity detection and resolution for text-to-sql},
  author={Ding, Zhongjun and Lin, Yin and Zeng, Tianjing and Zhu, Rong and Ding, Bolin and Zhou, Jingren},
  booktitle={Companion of the International Conference on Management of Data},
  pages={26--29},
  year={2026}
}

@InProceedings{pmlr-v267-shvartzshnaider25a,
  title = 	 {Position: Contextual Integrity is Inadequately Applied to Language Models},
  author =       {Shvartzshnaider, Yan and Duddu, Vasisht},
  booktitle = 	 {Proceedings of the 42nd International Conference on Machine Learning},
  pages = 	 {82200--82210},
  year = 	 {2025},
  editor = 	 {Singh, Aarti and Fazel, Maryam and Hsu, Daniel and Lacoste-Julien, Simon and Berkenkamp, Felix and Maharaj, Tegan and Wagstaff, Kiri and Zhu, Jerry},
  volume = 	 {267},
  series = 	 {Proceedings of Machine Learning Research},
  month = 	 {13--19 Jul},
  publisher =    {PMLR},
  url = 	 {https://proceedings.mlr.press/v267/shvartzshnaider25a.html}
}

@article{dp4sql,
  title={DP4SQL: Differentially Private SQL with Flexible Privacy Policies},
  author={Cascio, Andrew and Tong, KinChin and Kifer, Daniel and Ding, Zeyu and Zhang, Danfeng},
  journal={arXiv preprint arXiv:2606.07883},
  year={2026}
}

@inproceedings{blowfish,
  title={Blowfish privacy: Tuning privacy-utility trade-offs using policies},
  author={He, Xi and Machanavajjhala, Ashwin and Ding, Bolin},
  booktitle={Proceedings of the 2014 ACM SIGMOD international conference on Management of data},
  pages={1447--1458},
  year={2014}
}

@article{kifer2014pufferfish,
  title={Pufferfish: A framework for mathematical privacy definitions},
  author={Kifer, Daniel and Machanavajjhala, Ashwin},
  journal={ACM Transactions on Database Systems (TODS)},
  volume={39},
  number={1},
  pages={1--36},
  year={2014},
  publisher={ACM New York, NY, USA}
}

@inproceedings{chowdhury2025pr,
  title={Preempt: Sanitizing Sensitive Prompts for LLMs},
  author={Chowdhury, Amrita Roy and Glukhov, David and Anshumaan, Divyam and Chalasani, Prasad and Papernot, Nicolas and Jha, Somesh and Bellare, Mihir},
  booktitle={Network and Distributed System Security Symposium (NDSS)},
  year={2026},
  doi={10.14722/ndss.2026.231277}
}

@inproceedings{shvartzshanider-etal-2018-recipe,
    title = "{RECIPE}: Applying Open Domain Question Answering to Privacy Policies",
    author = "Shvartzshanider, Yan  and
      Balashankar, Ananth  and
      Wies, Thomas  and
      Subramanian, Lakshminarayanan",
    editor = "Choi, Eunsol  and
      Seo, Minjoon  and
      Chen, Danqi  and
      Jia, Robin  and
      Berant, Jonathan",
    booktitle = "Proceedings of the Workshop on Machine Reading for Question Answering",
    month = jul,
    year = "2018",
    address = "Melbourne, Australia",
    publisher = "Association for Computational Linguistics",
    url = "https://aclanthology.org/W18-2608/",
    doi = "10.18653/v1/W18-2608",
    pages = "71--77"
}

@inproceedings{shvartzshnaider2019going,
  title={Going against the (appropriate) flow: A contextual integrity approach to privacy policy analysis},
  author={Shvartzshnaider, Yan and Apthorpe, Noah and Feamster, Nick and Nissenbaum, Helen},
  booktitle={Proceedings of the AAAI Conference on Human Computation and Crowdsourcing},
  volume={7},
  pages={162--170},
  year={2019}
}

@inproceedings{shvartzshanider2023beyond,
    title = "Beyond The Text: Analysis of Privacy Statements through Syntactic and Semantic Role Labeling",
    author = "Shvartzshanider, Yan  and
      Balashankar, Ananth  and
      Wies, Thomas  and
      Subramanian, Lakshminarayanan",
    editor = "Preoțiuc-Pietro, Daniel  and
      Goanta, Catalina  and
      Chalkidis, Ilias  and
      Barrett, Leslie  and
      Spanakis, Gerasimos  and
      Aletras, Nikolaos",
    booktitle = "Proceedings of the Natural Legal Language Processing Workshop 2023",
    month = dec,
    year = "2023",
    address = "Singapore",
    publisher = "Association for Computational Linguistics",
    url = "https://aclanthology.org/2023.nllp-1.10/",
    pages = "85--98"
}

@article{chanenson2025automating,
  title={Automating governing knowledge commons and contextual integrity (GKC-CI) privacy policy annotations with large language models},
  author={Chanenson, Jake and Pickering, Madison and Apthorpe, Noah},
  journal={arXiv preprint arXiv:2311.02192},
  year={2023}
}

@inproceedings{xiao2024contextual,
    title = "Large Language Models Can Be Contextual Privacy Protection Learners",
    author = "Xiao, Yijia  and
      Jin, Yiqiao  and
      Bai, Yushi  and
      Wu, Yue  and
      Yang, Xianjun  and
      Luo, Xiao  and
      Yu, Wenchao  and
      Zhao, Xujiang  and
      Liu, Yanchi  and
      Gu, Quanquan  and
      Chen, Haifeng  and
      Wang, Wei  and
      Cheng, Wei",
    editor = "Al-Onaizan, Yaser  and
      Bansal, Mohit  and
      Chen, Yun-Nung",
    booktitle = "Proceedings of the 2024 Conference on Empirical Methods in Natural Language Processing",
    month = nov,
    year = "2024",
    address = "Miami, Florida, USA",
    publisher = "Association for Computational Linguistics",
    url = "https://aclanthology.org/2024.emnlp-main.785/",
    doi = "10.18653/v1/2024.emnlp-main.785",
    pages = "14179--14201"
}

@inproceedings{wang2025privacy,
    title = "Privacy in Action: Towards Realistic Privacy Mitigation and Evaluation for {LLM}-Powered Agents",
    author = "Wang, Shouju  and
      Yu, Fenglin  and
      Liu, Xirui  and
      Qin, Xiaoting  and
      Zhang, Jue  and
      Lin, Qingwei  and
      Zhang, Dongmei  and
      Rajmohan, Saravan",
    editor = "Christodoulopoulos, Christos  and
      Chakraborty, Tanmoy  and
      Rose, Carolyn  and
      Peng, Violet",
    booktitle = "Findings of the Association for Computational Linguistics: EMNLP 2025",
    month = nov,
    year = "2025",
    address = "Suzhou, China",
    publisher = "Association for Computational Linguistics",
    url = "https://aclanthology.org/2025.findings-emnlp.925/",
    doi = "10.18653/v1/2025.findings-emnlp.925",
    pages = "17055--17074",
    ISBN = "979-8-89176-335-7"
}

@misc{huang2026need,
      title={Need to Know: Contextual-Integrity-Grounded Query Rewriting for Privacy-Conscious LLM Delegation}, 
      author={Xinyue Huang and Xiaochun Cao and Wenyuan Yang},
      year={2026},
      eprint={2606.04067},
      archivePrefix={arXiv},
      primaryClass={cs.CR},
      url={https://arxiv.org/abs/2606.04067}, 
}

@inproceedings{hu2024differentially,
    title = "Differentially Private Natural Language Models: Recent Advances and Future Directions",
    author = "Hu, Lijie  and
      Habernal, Ivan  and
      Shen, Lei  and
      Wang, Di",
    editor = "Graham, Yvette  and
      Purver, Matthew",
    booktitle = "Findings of the Association for Computational Linguistics: EACL 2024",
    month = mar,
    year = "2024",
    address = "St. Julian{'}s, Malta",
    publisher = "Association for Computational Linguistics",
    url = "https://aclanthology.org/2024.findings-eacl.33/",
    doi = "10.18653/v1/2024.findings-eacl.33",
    pages = "478--499"
}

@inproceedings{meehan2022sentence,
    title = "Sentence-level Privacy for Document Embeddings",
    author = "Meehan, Casey  and
      Mrini, Khalil  and
      Chaudhuri, Kamalika",
    editor = "Muresan, Smaranda  and
      Nakov, Preslav  and
      Villavicencio, Aline",
    booktitle = "Proceedings of the 60th Annual Meeting of the Association for Computational Linguistics (Volume 1: Long Papers)",
    month = may,
    year = "2022",
    address = "Dublin, Ireland",
    publisher = "Association for Computational Linguistics",
    url = "https://aclanthology.org/2022.acl-long.238/",
    doi = "10.18653/v1/2022.acl-long.238",
    pages = "3367--3380"
}

@inproceedings{mattern2022limits,
    title = "The Limits of Word Level Differential Privacy",
    author = "Mattern, Justus  and
      Weggenmann, Benjamin  and
      Kerschbaum, Florian",
    editor = "Carpuat, Marine  and
      de Marneffe, Marie-Catherine  and
      Meza Ruiz, Ivan Vladimir",
    booktitle = "Findings of the Association for Computational Linguistics: NAACL 2022",
    month = jul,
    year = "2022",
    address = "Seattle, United States",
    publisher = "Association for Computational Linguistics",
    url = "https://aclanthology.org/2022.findings-naacl.65/",
    doi = "10.18653/v1/2022.findings-naacl.65",
    pages = "867--881"
}

@inproceedings{awon2025clusant,
    title = "{C}lu{S}an{T}: Differentially Private and Semantically Coherent Text Sanitization",
    author = "Awon, Ahmed Musa  and
      Lu, Yun  and
      Potka, Shera  and
      Thomo, Alex",
    editor = "Chiruzzo, Luis  and
      Ritter, Alan  and
      Wang, Lu",
    booktitle = "Proceedings of the 2025 Conference of the Nations of the Americas Chapter of the Association for Computational Linguistics: Human Language Technologies (Volume 1: Long Papers)",
    month = apr,
    year = "2025",
    address = "Albuquerque, New Mexico",
    publisher = "Association for Computational Linguistics",
    url = "https://aclanthology.org/2025.naacl-long.187/",
    doi = "10.18653/v1/2025.naacl-long.187",
    pages = "3676--3693",
    ISBN = "979-8-89176-189-6"
}

@inproceedings{zhang2025dyntext,
    title = "{DYNTEXT}: Semantic-Aware Dynamic Text Sanitization for Privacy-Preserving {LLM} Inference",
    author = "Zhang, Juhua  and
      Tian, Zhiliang  and
      Zhu, Minghang  and
      Song, Yiping  and
      Sheng, Taishu  and
      Yang, Siyi  and
      Du, Qiunan  and
      Liu, Xinwang  and
      Huang, Minlie  and
      Li, Dongsheng",
    editor = "Che, Wanxiang  and
      Nabende, Joyce  and
      Shutova, Ekaterina  and
      Pilehvar, Mohammad Taher",
    booktitle = "Findings of the Association for Computational Linguistics: ACL 2025",
    month = jul,
    year = "2025",
    address = "Vienna, Austria",
    publisher = "Association for Computational Linguistics",
    url = "https://aclanthology.org/2025.findings-acl.1038/",
    doi = "10.18653/v1/2025.findings-acl.1038",
    pages = "20243--20255",
    ISBN = "979-8-89176-256-5"
}

@inproceedings{mcmahan2018learning,title	= {Learning Differentially Private Recurrent Language Models},author	= {Brendan McMahan and Daniel Ramage and Kunal Talwar and Li Zhang},year	= {2018},URL	= {https://openreview.net/pdf?id=BJ0hF1Z0b},booktitle	= {International Conference on Learning Representations (ICLR)}}

@inproceedings{li2022large,
  title     = {{Large Language Models Can Be Strong Differentially Private Learners}},
  author    = {Li, Xuechen and Tramer, Florian and Liang, Percy and Hashimoto, Tatsunori},
  booktitle = {International Conference on Learning Representations},
  year      = {2022},
  url       = {https://mlanthology.org/iclr/2022/li2022iclr-large/}
}

@inproceedings{yu2022differentially,
author = {Yu, Da and Naik, Saurabh and Backurs, Arturs and Gopi, Sivakanth and Inan, Huseyin and Kamath, Gautam and Kulkarni, Janardhan (Jana) and Lee, Yin Tat and Manoel, Andre and Wutschitz, Lukas and Yekhanin, Sergey and Zhang, Huishuai},
title = {Differentially private fine-tuning of language models},
booktitle = {ICLR 2022},
year = {2022},
month = {April},
url = {https://www.microsoft.com/en-us/research/publication/differentially-private-fine-tuning-of-language-models/},
}

@inproceedings{anil2022large,
    title = "Large-Scale Differentially Private {BERT}",
    author = "Anil, Rohan  and
      Ghazi, Badih  and
      Gupta, Vineet  and
      Kumar, Ravi  and
      Manurangsi, Pasin",
    editor = "Goldberg, Yoav  and
      Kozareva, Zornitsa  and
      Zhang, Yue",
    booktitle = "Findings of the Association for Computational Linguistics: EMNLP 2022",
    month = dec,
    year = "2022",
    address = "Abu Dhabi, United Arab Emirates",
    publisher = "Association for Computational Linguistics",
    url = "https://aclanthology.org/2022.findings-emnlp.484/",
    doi = "10.18653/v1/2022.findings-emnlp.484",
    pages = "6481--6491"
}

@inproceedings{mattern2022language,
    title = "Differentially Private Language Models for Secure Data Sharing",
    author = {Mattern, Justus  and
      Jin, Zhijing  and
      Weggenmann, Benjamin  and
      Sch{\"o}lkopf, Bernhard  and
      Sachan, Mrinmaya},
    editor = "Goldberg, Yoav  and
      Kozareva, Zornitsa  and
      Zhang, Yue",
    booktitle = "Proceedings of the 2022 Conference on Empirical Methods in Natural Language Processing",
    month = dec,
    year = "2022",
    address = "Abu Dhabi, United Arab Emirates",
    publisher = "Association for Computational Linguistics",
    url = "https://aclanthology.org/2022.emnlp-main.323/",
    doi = "10.18653/v1/2022.emnlp-main.323",
    pages = "4860--4873"
}

@inproceedings{yue2023synthetic,
    title = "Synthetic Text Generation with Differential Privacy: A Simple and Practical Recipe",
    author = "Yue, Xiang  and
      Inan, Huseyin  and
      Li, Xuechen  and
      Kumar, Girish  and
      McAnallen, Julia  and
      Shajari, Hoda  and
      Sun, Huan  and
      Levitan, David  and
      Sim, Robert",
    editor = "Rogers, Anna  and
      Boyd-Graber, Jordan  and
      Okazaki, Naoaki",
    booktitle = "Proceedings of the 61st Annual Meeting of the Association for Computational Linguistics (Volume 1: Long Papers)",
    month = jul,
    year = "2023",
    address = "Toronto, Canada",
    publisher = "Association for Computational Linguistics",
    url = "https://aclanthology.org/2023.acl-long.74/",
    doi = "10.18653/v1/2023.acl-long.74",
    pages = "1321--1342"
}

@article{haney2016policy,
  title={Design of policy-aware differentially private algorithms},
  author={Haney, Samuel and Machanavajjhala, Ashwin and Ding, Bolin},
  journal={arXiv preprint arXiv:1404.3722},
  year={2014}
}

@inproceedings{shi2022selective,
    title = "Selective Differential Privacy for Language Modeling",
    author = "Shi, Weiyan  and
      Cui, Aiqi  and
      Li, Evan  and
      Jia, Ruoxi  and
      Yu, Zhou",
    editor = "Carpuat, Marine  and
      de Marneffe, Marie-Catherine  and
      Meza Ruiz, Ivan Vladimir",
    booktitle = "Proceedings of the 2022 Conference of the North American Chapter of the Association for Computational Linguistics: Human Language Technologies",
    month = jul,
    year = "2022",
    address = "Seattle, United States",
    publisher = "Association for Computational Linguistics",
    url = "https://aclanthology.org/2022.naacl-main.205/",
    doi = "10.18653/v1/2022.naacl-main.205",
    pages = "2848--2859"
}

@inproceedings{li-etal-2025-privaci,
    title = "{P}riva{CI}-Bench: Evaluating Privacy with Contextual Integrity and Legal Compliance",
    author = "Li, Haoran  and
      Hu, Wenbin  and
      Jing, Huihao  and
      Chen, Yulin  and
      Hu, Qi  and
      Han, Sirui  and
      Chu, Tianshu  and
      Hu, Peizhao  and
      Song, Yangqiu",
    editor = "Che, Wanxiang  and
      Nabende, Joyce  and
      Shutova, Ekaterina  and
      Pilehvar, Mohammad Taher",
    booktitle = "Proceedings of the 63rd Annual Meeting of the Association for Computational Linguistics (Volume 1: Long Papers)",
    month = jul,
    year = "2025",
    address = "Vienna, Austria",
    publisher = "Association for Computational Linguistics",
    url = "https://aclanthology.org/2025.acl-long.518/",
    doi = "10.18653/v1/2025.acl-long.518",
    pages = "10544--10559",
    ISBN = "979-8-89176-251-0"
}

@inproceedings{dong2022nearly,
  title={A nearly instance-optimal differentially private mechanism for conjunctive queries},
  author={Dong, Wei and Yi, Ke},
  booktitle={Proceedings of the 41st ACM SIGMOD-SIGACT-SIGAI Symposium on Principles of Database Systems},
  pages={213--225},
  year={2022}
}

@misc{sft,
      title={Unveiling the Secret Recipe: A Guide For Supervised Fine-Tuning Small LLMs}, 
      author={Aldo Pareja and Nikhil Shivakumar Nayak and Hao Wang and Krishnateja Killamsetty and Shivchander Sudalairaj and Wenlong Zhao and Seungwook Han and Abhishek Bhandwaldar and Guangxuan Xu and Kai Xu and Ligong Han and Luke Inglis and Akash Srivastava},
      year={2024},
      eprint={2412.13337},
      archivePrefix={arXiv},
      primaryClass={cs.LG},
      url={https://arxiv.org/abs/2412.13337}, 
}

@article{dpo,
  title={Direct preference optimization: Your language model is secretly a reward model},
  author={Rafailov, Rafael and Sharma, Archit and Mitchell, Eric and Manning, Christopher D and Ermon, Stefano and Finn, Chelsea},
  journal={Advances in neural information processing systems},
  volume={36},
  pages={53728--53741},
  year={2023}
}
\appendix
\newpage
\newpage
\section{Appendix}

\begin{table}[t]
\centering
\small
\setlength{\tabcolsep}{4pt}
\renewcommand{\arraystretch}{1.02}
\caption{Notation used throughout the paper.}
\label{tab:notation}
\begin{tabularx}{\columnwidth}{@{}lX@{}}
\toprule
\textbf{Notation} & \textbf{Meaning} \\
\midrule
$\NLQ,\Sch,\Pol$ & Natural language query; schema; privacy policy. \\
$\PrivDB,\SynDB$ & Private; DP synthetic database. \\
$\CandSet,\OptQ$ & Candidate queries; selected query. \\
$\Gamma$ & Schema integrity constraints. \\
\midrule
$I\sim I'$ & Neighboring database instances. \\
$\GS{q},\LS{q},\RS{q}$ & Global, local, residual sensitivity. \\
$q_E,\boundary{E}$ & Residual query; boundary attributes. \\
$\multiplicity{i}{b}$ & Boundary multiplicity at $b$. \\
$\LS{i},\witness{i}$ & Local sensitivity; maximizing boundary value. \\
$\mufree{1},\mufreemin,\Munocci$ & Private boundary multiplicities. \\
$\PrivRels,\DistSet{k}$ & Private relations; distance vectors. \\
\midrule
$\Vsnd{i},\Vrec{i},\Vsub{i}$ & Sender, recipient, subject views. \\
$\Vth{i},\Vdat{i}$ & Condition and data views. \\
$o=(R,\mathit{tid},A)$ & Source cell provenance. \\
$\Req=(s_\rho,r_\rho,\pi_\rho)$ & Runtime sender, recipient, purpose. \\
$\Pos_\rho,\Viol_\rho$ & Permitted and prohibited information. \\
$\Vex$ & Information authorized for exact release. \\
$\public$ & Authorized database cells. \\
\bottomrule
\end{tabularx}
\end{table}

\subsection{Integrating Contextual Integrity and Differential Privacy (Cntd.)}
\label{app:cidp}

We discuss our inference rules in detail below:

\paragraph{Boundary multiplicities.}
Recall from Sec.~\ref{sec:background} that RS measures how strongly a tuple in $R_i$ can
be amplified through the residual join. For $R_i$, this amplification is
determined by the largest boundary multiplicity,
$\witness{i}=\LS{i}$. Thus, contextual information can tighten the local
sensitivity only if it can remove the boundary value
realizing this witness from the neighborhood.

Under contextual DP, values authorized by $\public$ are fixed across
contextual neighbors. Among the occupied boundary values that remain
private, let $\mufree{1}$ and $\mufreemin$ denote the largest and smallest
multiplicities, respectively. Under deletion, $\mufree{1}$ is the largest
private contribution that can still be removed. Under change neighborhood, a
tuple may move from one boundary value to another, so $\mufreemin$ captures
the smallest contribution it can move to; changing from the original witness
can therefore change the answer by at most
$\witness{i}-\mufreemin$. Hence,
$\mufreemin\leq\mufree{1}\leq\witness{i}$.

For insertions, $\Munocci$ denotes the largest multiplicity at a boundary
value not currently occupied by $R_i$, since a new tuple may introduce such
a value. If referential integrity prevents introducing new boundary values,
then $\Munocci=0$. The rules below characterize which of these multiplicities
remain reachable under each neighboring edit.

\noindent\textbf{Deletion.}
When a selection of tuples are contextually-public, i.e., fixed across contextual
neighbors, they cannot be deleted. Any valid deletion must come from
the remaining private tuples. The worst case is thus the private tuple with
the largest boundary multiplicity, giving $\mufree{1}$ in
\textsc{Sel-Del}. Intuitively, contextual information removes public tuples
from the set of possible deletion witnesses.

\noindent\textbf{Change.}
Under change neighborhood, tuple existence is fixed, but an existing tuple may
change the boundary value through which it joins with the residual query.
If the complete boundary projection is public, these join relevant values
are fixed, so changing the tuple cannot alter its join multiplicity and the
sensitivity contribution is zero (\textsc{Prj-Chg}). If only part of the
boundary information is public, a change can still have a large effect in
two ways: it can move a tuple away from the original witness, changing its
contribution by at most $\witness{i}-\mufreemin$, or it can move to a fully
private boundary value with multiplicity at most $\mufree{1}$. Taking the
larger of these possibilities gives
$\max\{\witness{i}-\mufreemin,\mufree{1}\}$ in
\textsc{PrjSel-Chg} and in the model allowing both deletion and change
(\textsc{Sel}).

\noindent\textbf{Insertion.}
Insertion is different because a new tuple is not restricted to boundary
values already present in the public data. It may introduce a new boundary
value and still realize the ordinary worst case, so \textsc{Ins} retains
$\LS{i}$. A tighter bound becomes possible when $\boundary{i}$ is unique:
with a unique boundary, an insertion cannot reuse an occupied boundary value,
so its effect is bounded by $\Munocci$, the largest multiplicity at an
unoccupied boundary value. Deletion or change may still affect an existing
private tuple, with contribution bounded by $\mufree{1}$. Since the rule
must cover all allowed edits, \textsc{Sel-Unique} gives
$\max\{\mufree{1},\Munocci\}$. Referential integrity
can further rule out new boundary values, forcing $\Munocci=0$. Finally, if
the entire relation is public, every tuple and value is fixed across
contextual neighbors, so \textsc{Pub} gives zero sensitivity.
Unlike ordinary residual sensitivity, contextual residual sensitivity is query- and context-dependent: the same query on the same database can have different sensitivity for different queriers or purposes because CI authorizes different information flows.

\subsection{Proofs}
\label{app:proofs}
\begin{theorem}[End-to-end differential privacy]
\label{thm:end-to-end-dp}
If the synthetic release is $(\varepsilon_{\mathrm{syn}},\delta_{\mathrm{syn}})$-DP and the final query release is $(\varepsilon_{\mathrm{rel}},\delta_{\mathrm{rel}})$-DP, then \sysname{} is
$(\varepsilon_{\mathrm{syn}}+\varepsilon_{\mathrm{rel}},
\delta_{\mathrm{syn}}+\delta_{\mathrm{rel}})$-DP.
\end{theorem}

\begin{proof}
\sysname{} accesses $\PrivDB$ only twice. First, it releases $\SynDB$ with
$(\varepsilon_{\mathrm{syn}},\delta_{\mathrm{syn}})$-DP. The entire search,
including sensitivity evaluation and selection of $\OptQ$, operates only on
$\SynDB$ and is therefore post-processing. Second, \sysname{} evaluates
$\OptQ$ on $\PrivDB$ using an
$(\varepsilon_{\mathrm{rel}},\delta_{\mathrm{rel}})$-DP release mechanism.
Although $\OptQ$ depends on $\SynDB$, conditioned on $\SynDB$ it is fixed.
The result therefore follows by adaptive composition~\cite{dwork2014algorithmic}.
\end{proof}
\begin{theorem}[Equivalence of CI Policy and Compiled Views]
\label{thm:ci-equivalence}
Fix a database $D$ and query
$\rho=(s_\rho,r_\rho,\pi_\rho,q_\rho)$.
Let $\Pol^{+}$ denote the implication-completed CI policy, and assume
every norm in $\Pol^{+}$ is faithfully compiled into its SQL views.

Let
\[
\llbracket \Pol^{+} \rrbracket_{D,\rho}
=
\left\{
\begin{array}{l}
(v,o)\in q_\rho(D)\text{ such that}\\
\text{$(v,o)$ is permitted by at least one}\\
\text{applicable positive norm of $\Pol^{+}$ and}\\
\text{prohibited by no applicable negative}\\
\text{norm of $\Pol^{+}$}
\end{array}
\right\}
\]
denote the information authorized by the CI policy for $\rho$.
Then the compiled SQL views are equivalent to the CI policy:
\[
\boxed{
\Vex
=
\Pos_\rho\setminus\Viol_\rho
=
\llbracket \Pol^{+} \rrbracket_{D,\rho}.
}
\]

Equivalently, for every origin-tagged queried value
$(v,o)\in q_\rho(D)$,
\[
(v,o)\in\Vex
\quad\Longleftrightarrow\quad
(v,o)\in\llbracket \Pol^{+}\rrbracket_{D,\rho}.
\]
Thus, CI enforcement is both sound and complete: it releases no queried
value prohibited by the policy and withholds no queried value authorized
by the policy. Consequently, releasing $\Vex$ satisfies contextual
integrity with respect to $\Pol$.
\end{theorem}

\begin{proof}
Fix an origin-tagged queried value $(v,o)\in q_\rho(D)$, where
$o=(R,\mathit{tid},A)$. We show that $(v,o)$ belongs to $\Vex$ if and
only if its flow is authorized by $\Pol^{+}$.

First, runtime relevance filtering does not remove any norm that can
govern $(v,o)$. A norm $N_i$ is discarded only if its data view
$\Vdat{i}$ governs none of the attributes projected by $q_\rho$.
By faithful compilation, $\Vdat{i}$ contains exactly the database cells
representing information of type $\tau_i$. Hence, if $N_i$ governs
$(v,o)$, then its origin attribute $A$ is represented by $\Vdat{i}$ and
$N_i$ is retained. Conversely, a discarded norm governs no queried
value and therefore cannot affect whether $(v,o)$ belongs to
$\Pos_\rho$ or $\Viol_\rho$.

Now consider any retained norm $N_i$. At runtime, \sysname{} evaluates
\[
\pi_{q,v,o}
\left(
\sigma_{p_1=s_\rho}(\Vsnd{i})
\Join
\sigma_{p_2=r_\rho}(\Vrec{i})
\Join
\Vsub{i}
\Join
\Vth{i}
\Join
\Vdat{i}
\right).
\]
A tuple $(q,v,o)$ appears in this result if and only if there exists a
joint binding of the shared variables $(p_1,p_2,q)$, with
$p_1=s_\rho$ and $p_2=r_\rho$, under which all compiled views hold.
The natural join therefore implements conjunction of the corresponding
conditions.

By faithful compilation, the role views hold exactly when
$s_\rho$, $r_\rho$, and $q$ instantiate the sender, recipient, and
subject roles of $N_i$. Likewise, $\Vth{i}$ holds exactly when the
transmission condition $\theta_i$ is satisfied and the declared purpose
$\pi_\rho$ matches $\pi_i$, while $\Vdat{i}$ contains exactly the
origin-tagged values of information type $\tau_i$ about the bound
subject. Therefore,
\[
\begin{aligned}
&(q,v,o)\text{ is returned by the compiled views of }N_i\\
&\quad\Longleftrightarrow\quad
N_i\text{ applies to the flow of }(v,o)\text{ under }\rho.
\end{aligned}
\]

Taking the union over all applicable positive norms therefore gives
exactly
\[
\Pos_\rho
=
\left\{
(v,o)\in q_\rho(D)
:
\begin{array}[t]{l}
\text{some positive norm of $\Pol^{+}$}\\
\text{permits $(v,o)$}
\end{array}
\right\},
\]
and, by the same argument,
\[
\Viol_\rho
=
\left\{
(v,o)\in q_\rho(D)
:
\begin{array}[t]{l}
\text{some negative norm of $\Pol^{+}$}\\
\text{prohibits $(v,o)$}
\end{array}
\right\}.
\]

It follows directly that
\[
\begin{aligned}
(v,o)\in\Vex
&\Longleftrightarrow
(v,o)\in\Pos_\rho\setminus\Viol_\rho \\
&\Longleftrightarrow
\text{$(v,o)$ is permitted by at least one}\\
&\hspace{1.7cm}\text{applicable positive norm and}\\
&\hspace{1.7cm}\text{prohibited by no applicable}\\
&\hspace{1.7cm}\text{negative norm}\\
&\Longleftrightarrow
(v,o)\in\llbracket\Pol^{+}\rrbracket_{D,\rho}.
\end{aligned}
\]
Since this holds for every $(v,o)\in q_\rho(D)$,
\[
\Vex=\llbracket\Pol^{+}\rrbracket_{D,\rho}.
\]

Finally, $\Pol^{+}$ preserves the permissions of $\Pol$ and adds only
prohibitions induced by information implications. Hence every flow
authorized under $\Pol^{+}$ is compliant with $\Pol$, and releasing
$\Vex$ satisfies contextual integrity with respect to $\Pol$.
\end{proof}

\bigskip
\noindent\textbf{Contextual residual sensitivity.}
We now prove Theorem~\ref{thm:crs} of Sec.~\ref{sec:ci-dp}: contextual
residual sensitivity ($\CRS{q}$) is a smooth upper bound on local
sensitivity under contextual DP, and it never exceeds ordinary residual
sensitivity ($\RS{q}$). Throughout, fix a query $q$ and a context $c$ with
authorized cells $\public$.
\smallskip
\noindent\textbf{Setup.}
For each private relation $R_i$, write $\ell_i(I)$ for the bound given by
the inference rule of Sec.~\ref{sec:ci-dp} that applies to $R_i$ under
$\public$, i.e., $\Gamma,\public\vdash_I R_i:\tau\Rightarrow\ell_i(I)$.
Every contextual neighbor must preserve $\public$
(Definition~\ref{def:contextual-dp}), so a relation entirely governed by
\textsc{Pub} admits no contextual edit at all: inserting, deleting, or
changing any of its tuples would change $\public$. We call the remaining
relations \emph{contextually editable},
\[
\PrivRelsC=\{i\in\PrivRels:\ell_i\text{ is not derived by \textsc{Pub}}\},
\]
and restrict distance vectors to them,
$\DistSet{k}^{c}=\{\mathbf{s}\in\DistSet{k}:s_j=0\text{ for all }
j\notin\PrivRelsC\}$. The \emph{contextual distance} $d_c(I,I')$ is the
fewest edits needed to reach $I'$ from $I$ while preserving $\public$ at
every intermediate step. Since a contextual edit is in particular an
ordinary edit, $d(I,I')\le d_c(I,I')$, and any instance within contextual
distance $k$ of $I$ can only differ from $I$ on relations in $\PrivRelsC$.
The \emph{contextual local sensitivity} of $q$ at $I$ is local sensitivity
restricted to contextual neighbors,
\[
\CLS{q,\public}(I)=\max_{I':I\sim_c I'}\big|\,|q(I)|-|q(I')|\,\big|.
\]
As with ordinary DP, we cannot calibrate noise to $\CLS{q,\public}$ directly, since
it can itself change sharply between neighbors. Following Nissim et
al.~\cite{nissim2007smooth}, a function $\hat{S}(\cdot)$ is a
\emph{$\beta$-smooth upper bound of local sensitivity under contextual DP}
if (i) $\hat{S}(I)\ge\CLS{q,\public}(I)$ for every $I$, and
(ii) $\hat{S}(I)\le e^{\beta}\hat{S}(I')$ for every $I\sim_c I'$; noise
calibrated to any such $\hat{S}$ using their Cauchy or Laplace mechanisms
achieves $\varepsilon$- or $(\varepsilon,\delta)$-contextual DP.

\begin{definition}[Contextual residual sensitivity]
\label{def:crs}
For $k\ge 0$, the contextual distance-$k$ envelope $\CLShat_q^{(k)}$ is
defined exactly as the envelope of Eq.~\eqref{eq:ls-distance}, but with
distance vectors and boundary maxima restricted to $\PrivRelsC$:
\[
\CLShat_q^{(k)}(I)
=
\begin{cases}
\displaystyle
\max_{i\in\PrivRelsC}\ \ell_i(I),
& k=0,\\[8pt]
\displaystyle
\max_{i\in\PrivRelsC}\
\max_{\mathbf{s}\in\DistSet{k}^{c}}\
\sum_{A\subseteq E_i\cap\PrivRelsC}
T_{E_i\setminus A}(I)\prod_{j\in A}s_j,
& k\ge 1.
\end{cases}
\]
The \emph{contextual residual sensitivity} of $q$ on $I$ is then
\[
\CRS{q}(I)
=
\max_{k\ge 0}\
e^{-\beta k}
\min\big\{\GS{q},\ \CLShat_q^{(k)}(I)\big\}.
\]
\end{definition}
As with $\RS{q}$ (Eq.~\eqref{eq:residual-sensitivity}), an unbounded
$\GS{q}$ is treated as $\infty$. If $\PrivRelsC=\emptyset$, every relation
is fully authorized, so there is no contextual neighbor left to protect
against: $\CLS{q,\public}\equiv 0$ and $\CRS{q,\public}\equiv 0$, and the exact answer may
be released.

We prove Theorem~\ref{thm:crs} through three lemmas, each a contextual
counterpart of a result of Dong and Yi~\cite{dong2021residual}: the
inference rules are \emph{sound} (Lemma~\ref{lem:rule-sound}, analogous to
their Theorem 4.5), the rules are \emph{tight} against $T_{E_i}$
(Lemma~\ref{lem:rule-tight}), and the envelope $\CLShat_q^{(k)}$ is
\emph{smooth} across contextual neighbors (Lemma~\ref{lem:crs-smooth},
analogous to their Lemma 4.11).

\begin{lemma}[Rule soundness]
\label{lem:rule-sound}
For every instance $I$,
$\CLS{q,\public}(I)\le\max_{i\in\PrivRelsC}\ell_i(I)=\CLShat_q^{(0)}(I)$.
\end{lemma}

\begin{proof}
Let $I\sim_c I'$, so $I'$ is obtained from $I$ by a single edit $\tau$ to
some private relation $R_i$, and this edit preserves $\public$. Were
$\ell_i$ derived by \textsc{Pub}, every cell of $R_i$ would be authorized
and any edit would change $\public$; hence $i\in\PrivRelsC$.

Dong and Yi show that a single edit to $R_i$ changes the query answer by
exactly the multiplicity of the affected boundary value: if the edit
inserts or deletes a tuple with boundary value $b$, then
$||q(I)|-|q(I')||=\multiplicity{i}{b}$; if it instead changes an existing
tuple's boundary value from $b$ to $b'$ (non-boundary attributes do not
affect the count), then
$||q(I)|-|q(I')||=|\multiplicity{i}{b}-\multiplicity{i}{b'}|$. It remains
to bound the multiplicities each rule allows.

\textbf{Base, Ins.} No tuple of $R_i$ is fixed by $\public$, so $b$ ranges
over the whole domain of $\boundary{i}$, and the change is at most
$\max_b\multiplicity{i}{b}=T_{E_i}(I)=\LS{i}=\ell_i$.

\textbf{Sel-Del.} Deleting an authorized tuple would remove a cell from
$\public$, so the deleted tuple must be unauthorized, meaning $b$ is a
boundary value occupied by a private tuple. The change is thus at most
$\mufree{1}=\ell_i$.

\textbf{Prj-Chg.} Every boundary cell of $R_i$ is authorized and hence
fixed across contextual neighbors, so a change can only touch non-boundary
attributes: $b'=b$, the answer is unchanged, matching $\ell_i=0$.

\textbf{PrjSel-Chg, Sel.} The changed or deleted tuple is unauthorized, so
its boundary value $b$ is privately occupied,
$\mufreemin\le\multiplicity{i}{b}\le\mufree{1}$; a change may move it to
any boundary value $b'$, with $\multiplicity{i}{b'}\le\witness{i}$.
Whichever of $b,b'$ has larger multiplicity, the change is at most
$\max\{\witness{i}-\mufreemin,\mufree{1}\}=\ell_i$.

\textbf{Sel-Unique.} Because $\boundary{i}$ is unique, an inserted tuple
(or the new value of a changed tuple) must occupy a boundary value not
already present in $R_i$, whose multiplicity is at most $\Munocci$; a
deleted or changed tuple is instead unauthorized and privately occupied,
so its multiplicity is at most $\mufree{1}$. Either way, the change is at
most $\max\{\mufree{1},\Munocci\}=\ell_i$.

Taking the maximum over $i\in\PrivRelsC$ gives
$\CLS{q,\public}(I)\le\max_{i\in\PrivRelsC}\ell_i(I)=\CLShat_q^{(0)}(I)$.
\end{proof}

\begin{lemma}[Rule tightness]
\label{lem:rule-tight}
For every rule and every instance $I$,
$\ell_i(I)\le T_{E_i}(I)=\LS{i}$.
\end{lemma}

\begin{proof}
\textsc{Base} and \textsc{Ins} return $\LS{i}$ exactly. Every other bound
is a maximum of $\multiplicity{i}{b}$ over some subset of
$\operatorname{dom}(\boundary{i})$, whereas $\LS{i}=T_{E_i}(I)$ maximizes
over the entire domain (Eq.~\eqref{eq:te}), and a maximum over a subset
can only be smaller. Concretely, $\mufree{1}\le\LS{i}$ and
$\Munocci\le\LS{i}$ cover \textsc{Sel-Del} and \textsc{Sel-Unique}; since
$\mufreemin\ge 0$, $\witness{i}-\mufreemin\le\witness{i}=\LS{i}$ covers
\textsc{PrjSel-Chg} and \textsc{Sel}; and $0\le\LS{i}$ covers
\textsc{Prj-Chg} and \textsc{Pub}.
\end{proof}

\begin{lemma}[Smoothness]
\label{lem:crs-smooth}
For every $k\ge 0$ and every pair $I\sim_c I'$,
$\CLShat_q^{(k)}(I)\le\CLShat_q^{(k+1)}(I')$.
\end{lemma}

\begin{proof} We use two properties of $T_E$ established by Dong and Yi~\cite{dong2021residual}. First, $T_E(I)$ depends only on the relations indexed by $E$ and is therefore invariant under edits to relations outside $E$. Second, for any $j\in E$, a single edit to $R_j$ changes $T_E(I)$ by at most $T_{E\setminus\{j\}}(I)$. These properties allow the residual-sensitivity envelope to account for neighboring databases by incrementing the corresponding coordinates of its distance vector. The same argument applies when the ground set is restricted to $\PrivRelsC$: only coordinates corresponding to editable relations can be incremented, and Lemma~\ref{lem:rule-sound} establishes that every such relation belongs to $\PrivRelsC$. Thus, the residual-sensitivity construction carries over with both its ground set and distance vector restricted to $\PrivRelsC$.

Let $I\sim_c I'$ differ in relation $R_j$, so $j\in\PrivRelsC$, and write
$T^{c}_{E,\mathbf{s}}(I)=\sum_{A\subseteq E\cap\PrivRelsC}
T_{E-A}(I)\prod_{i\in A}s_i$ for the weighted sum defining
$\CLShat_q^{(k)}$.

\emph{Case $k=0$.} Let $i\in\PrivRelsC$ attain
$\CLShat_q^{(0)}(I)=\ell_i(I)$, which is at most $T_{E_i}(I)$ by
Lemma~\ref{lem:rule-tight}. If $j=i$, fact (a) gives
$T_{E_i}(I)=T_{E_i}(I')$ directly, since $i\notin E_i$. Otherwise
$j\in E_i\cap\PrivRelsC$, and fact (b) gives
$T_{E_i}(I)\le T_{E_i}(I')+T_{E_i-\{j\}}(I')$. In both cases the
right-hand side is a sum of terms of $T^{c}_{E_i,\mathbf{e}_j}(I')$ for the
unit vector $\mathbf{e}_j\in\DistSet{1}^{c}$, so
$\CLShat_q^{(0)}(I)\le\CLShat_q^{(1)}(I')$.

\emph{Case $k\ge 1$.} Let $(i,\mathbf{s})$ attain $\CLShat_q^{(k)}(I)$, and
set $\mathbf{s}'=\mathbf{s}+\mathbf{e}_j\in\DistSet{k+1}^{c}$. By the
smoothness argument above, applied to the ground set
$E_i\cap\PrivRelsC$, incrementing coordinate $j$ can only increase the sum:
$T^{c}_{E_i,\mathbf{s}}(I)\le T^{c}_{E_i,\mathbf{s}'}(I')$. Hence
$\CLShat_q^{(k)}(I)\le T^{c}_{E_i,\mathbf{s}'}(I')\le\CLShat_q^{(k+1)}(I')$.
\end{proof}

\begin{theorem}[Contextual residual sensitivity]
\label{thm:crs-app}
For every multiway join counting query $q$, instance $I$, context $c$, and
smoothing parameter $\beta>0$:
\begin{enumerate}[label=(\roman*),leftmargin=*]
\item $\CRS{q}(I)\le\RS{q}(I)$; and
\item $\CRS{q}(\cdot)$ is a $\beta$-smooth upper bound of local sensitivity
under contextual DP.
\end{enumerate}
Consequently, releasing $|q(I)|$ with general Cauchy (resp.\ Laplace) noise
calibrated to $\CRS{q}(I)$, with $\beta$ set exactly as for residual
sensitivity, satisfies $\varepsilon$- (resp.\ $(\varepsilon,\delta)$-)
contextual differential privacy, and the injected noise never exceeds that
of standard residual sensitivity.
\end{theorem}

\begin{proof}
\emph{(i) $\CRS{q}$ never exceeds $\RS{q}$.} We show
$\CLShat_q^{(k)}(I)\le\LS{q}^{(k)}(I)$ for every $k$; since $\CRS{q}(I)$
and $\RS{q}(I)$ apply the same monotone transformation
$x\mapsto e^{-\beta k}\min\{\GS{q},x\}$ to these envelopes and maximize
over the same range of $k$ (Eq.~\eqref{eq:residual-sensitivity}), this
gives $\CRS{q}(I)\le\RS{q}(I)$.

For $k=0$, Lemma~\ref{lem:rule-tight} gives
$\CLShat_q^{(0)}(I)=\max_{i\in\PrivRelsC}\ell_i(I)
\le\max_{i\in\PrivRels}T_{E_i}(I)=\LS{q}(I)$.
For $k\ge1$, the pair $(i,\mathbf{s})$ attaining $\CLShat_q^{(k)}(I)$
satisfies $\mathbf{s}\in\DistSet{k}^{c}\subseteq\DistSet{k}$, and its
defining sum ranges only over $A\subseteq E_i\cap\PrivRelsC$ --- a subset
of the nonnegative terms summed by $\LS{q}^{(k)}(I)$, which ranges over
$A\subseteq E_i\cap\PrivRels$ (Eq.~\eqref{eq:ls-distance}). Dropping terms
from a sum of nonnegative quantities cannot increase it, so
$\CLShat_q^{(k)}(I)\le\LS{q}^{(k)}(I)$.

\emph{(ii) $\CRS{q}$ is a smooth upper bound under contextual DP.} We
check the two defining conditions. For condition (i),
\[
\CRS{q,\public}(I)\ge\min\big\{\GS{q},\CLShat_q^{(0)}(I)\big\}\ge\CLS{q,\public}(I),
\]
since $\CLShat_q^{(0)}(I)\ge\CLS{q,\public}(I)$ by Lemma~\ref{lem:rule-sound}, and
$\GS{q}\ge\LS{q}(I)\ge\CLS{q,\public}(I)$ because every contextual neighbor is also
a standard neighbor.

For condition (ii), fix $I\sim_c I'$ and let $k^{\star}$ attain the
maximum defining $\CRS{q}(I)$; such a maximizer exists at some finite
$k^{\star}$, exactly as for standard residual sensitivity (Dong and Yi,
Lemma 4.12), with the same argument applying verbatim over $\PrivRelsC$.
By Lemma~\ref{lem:crs-smooth} and then rescaling by $e^{\beta}$,
\begin{align*}
\CRS{q}(I)
&=e^{-\beta k^{\star}}
\min\big\{\GS{q},\CLShat_q^{(k^{\star})}(I)\big\}\\
&\le e^{-\beta k^{\star}}
\min\big\{\GS{q},\CLShat_q^{(k^{\star}+1)}(I')\big\}\\
&=e^{\beta}\cdot e^{-\beta(k^{\star}+1)}
\min\big\{\GS{q},\CLShat_q^{(k^{\star}+1)}(I')\big\}\\
&\le e^{\beta}\CRS{q}(I').
\end{align*}

Since $\CRS{q}$ satisfies both conditions with respect to $\sim_c$, the
Cauchy and Laplace mechanisms of~\cite{nissim2007smooth}, calibrated to
$\CRS{q}(I)$ with the same $\beta$ used for residual sensitivity, achieve
$\varepsilon$- and $(\varepsilon,\delta)$-contextual DP, respectively.
Part (i) then implies this noise is never larger than what standard
residual sensitivity would require.
\end{proof}

\bigskip
\noindent\textbf{End-to-end privacy.}
The following proof establishes Corollary~\ref{cor:end-to-end} of
Sec.~\ref{sec:ci-dp}.

\begin{proof}
Write the end-to-end mechanism as
\[
\mathcal{M}(\PrivDB)=\big(\SynDB,\ \mathcal{M}_{\mathrm{rel}}(\PrivDB;\OptQ)\big),
\qquad
\SynDB=\SynMech(\PrivDB),\quad\OptQ=f(\SynDB),
\]
where $f$ is the query-selection procedure: a function of $\SynDB$ and
public inputs (the natural-language question and schema) alone, which
never reads $\PrivDB$.

By hypothesis, $\SynDB$ is $(\varepsilon_{\mathrm{syn}},
\delta_{\mathrm{syn}})$-DP. Since $\OptQ=f(\SynDB)$ depends only on
$\SynDB$, the pair $(\SynDB,\OptQ)$ is a post-processing of $\SynDB$, and
post-processing preserves DP under any neighboring relation
(Prop.~2.1 of~\cite{dwork2014algorithmic}), in particular the contextual
relation $\sim_c$. So $(\SynDB,\OptQ)$ remains
$(\varepsilon_{\mathrm{syn}},\delta_{\mathrm{syn}})$-DP.

By hypothesis, $\mathcal{M}_{\mathrm{rel}}(\,\cdot\,;q)$ is
$(\varepsilon_{\mathrm{rel}},\delta_{\mathrm{rel}})$-contextual DP for
every query $q$ in the DP-compatible class. Since $\OptQ$ always lies in
this class, this guarantee holds in particular for $q=\OptQ$, regardless
of which query is realized.

Together, these two facts say that $\mathcal{M}$ is an adaptive
composition of two mechanisms: $\mathcal{M}_1(D)=(\SynDB,\OptQ)$, which is
$(\varepsilon_{\mathrm{syn}},\delta_{\mathrm{syn}})$-DP, followed by
$\mathcal{M}_2(D;z)=\mathcal{M}_{\mathrm{rel}}(D;z)$, which is
$(\varepsilon_{\mathrm{rel}},\delta_{\mathrm{rel}})$-contextual DP for
every fixed $z$. Contextual DP is exactly $(\varepsilon,\delta)$-DP with
respect to the fixed neighboring relation $\sim_c$
(Definition~\ref{def:contextual-dp}), so the ordinary sequential
composition theorem for $(\varepsilon,\delta)$-DP (Thm.~3.16
of~\cite{dwork2014algorithmic}) applies with $\sim_c$ in place of $\sim$,
giving that $\mathcal{M}$ is
$(\varepsilon_{\mathrm{syn}}+\varepsilon_{\mathrm{rel}},
\delta_{\mathrm{syn}}+\delta_{\mathrm{rel}})$-contextual DP.

Theorem~\ref{thm:crs} shows that calibrating $\mathcal{M}_{\mathrm{rel}}$
to $\CRS{q}$ satisfies this hypothesis for every $q$ in the DP-compatible
class, so \sysname{} instantiated this way achieves the guarantee above. Finally, any additional output that depends only on $\widetilde D$ and
public inputs, such as the SQL text of $q^\star$, incurs no additional
privacy cost by post-processing. The authorization boundary
$\mathcal{A}_c(D)$ is instead computed from the private database $D$ under
the fixed context $c$. However, by definition of contextual adjacency,
$\mathcal{A}_c(D)=\mathcal{A}_c(D')$ for every $D\sim_c D'$.
Thus, $\mathcal{A}_c(D)$ is invariant across contextual neighbors and its
use in defining $\sim_c$ or calibrating the release incurs no additional
contextual-DP loss.
\end{proof}
\subsection{Procedures}\label{app:alg}
\begin{mybox2}{\algotitle{alg:dp-mcts}{\sysname: DP-Aware Text2SQL}}
\footnotesize
\setlength{\parindent}{0pt}

\textbf{Input:} question $\NLQ$, schema $\Sch$, private database $\PrivDB$,
rollout budget $N_{\mathrm{rollout}}$, warm-start size $K$

\textbf{Output:} selected query $\OptQ$

\vspace{1mm}
\textit{Initialization}

\begin{enumerate}[nosep, label=\arabic*.]

\item $\SynDB \gets
\SynMech(\PrivDB;\varepsilon_{\mathrm{syn}},\delta_{\mathrm{syn}})$
\hfill
\textcolor{blue}{$\rhd$ One-time private proxy}

\item Initialize root $v_0$ with
$\{\ActSetCount$,$\ActAddRelation$,$\ActStop\}$
\hfill
\textcolor{blue}{$\rhd$ DP-compatible tree}

\item Set $N(v)\gets0$, $N(v,a)\gets0$, $Q(v,a)\gets0$
\hfill
\textcolor{blue}{$\rhd$ Node visits, edge visits, cumulative reward}

\item Set $\{C_y\}\gets\emptyset$, $\{\Incumbent{y}\}\gets\emptyset$
\hfill
\textcolor{blue}{$\rhd$ Answer clusters + best RS per cluster}

\item Generate $K$ zero-shot queries and insert their action paths
\hfill
\textcolor{blue}{$\rhd$ Warm start}

\item \textbf{For} each warm-start query $q$

\item \hspace{0.25cm}
$y\gets\operatorname{bucket}(q(\SynDB))$;
$C_y\gets C_y\cup\{q\}$
\hfill
\textcolor{blue}{$\rhd$ $y$: answer; $C_y$: queries with answer $y$}

\item \hspace{0.25cm}
Compute $\ResRS{q}$ and set
$\Incumbent{y}\gets
\min_{q'\in C_y}\ResRS{q'}$
\hfill
\textcolor{blue}{$\rhd$ Lowest RS found for answer $y$}

\item \hspace{0.25cm}
Evaluate and backpropagate $q$ along its path
\hfill
\textcolor{blue}{$\rhd$ Initialize $Q$ and $N$}

\item \textbf{EndFor}

\end{enumerate}

\vspace{1mm}
\textit{Search and Evaluation}

\begin{enumerate}[nosep, label=\arabic*.]
\setcounter{enumi}{10}

\item \textbf{For} rollout $i\in[N_{\mathrm{rollout}}]$

\item \hspace{0.25cm}
$\textsc{Rollout}(v_0,i,\SynDB)$
\hfill
\textcolor{blue}{$\rhd$ Explore + evaluate + update}

\item \textbf{EndFor}

\end{enumerate}

\vspace{1mm}
\textit{Final Selection}

\begin{enumerate}[nosep, label=\arabic*.]
\setcounter{enumi}{13}

\item $y^\star\gets\arg\max_y |C_y|$
\hfill
\textcolor{blue}{$\rhd$ Most supported answer}

\item $\OptQ\gets\arg\min_{q\in C_{y^\star}}\ResRS{q}$
\hfill
\textcolor{blue}{$\rhd$ Lowest-RS formulation}

\item \textbf{Return} $\OptQ$

\end{enumerate}

\end{mybox2}
\newpage
\begin{mybox2}{\algotitle{alg:dp-rollout}{\sysname: Rollout}}
\footnotesize
\setlength{\parindent}{0pt}

\textbf{Input:} root $v_0$, rollout $i$, synthetic database $\SynDB$

\vspace{1mm}
\textit{Search}

\begin{enumerate}[nosep, label=\arabic*.]

\item Select from $v_0$ using the UCT rule until reaching $v$ with
      $|\mathrm{children}(v)| < k\sqrt{N(v)}$
\hfill
\textcolor{blue}{$\rhd$ More visited nodes are allowed more branches}

\item Generate DP-compatible child actions
\hfill
\textcolor{blue}{$\rhd$ LLM expansion}

\item Exclude relation keys already used by siblings
\hfill
\textcolor{blue}{$\rhd$ Remove duplicate branches}

\item \textbf{For} each
$a=\ActAddRelation(r,J,F)$ producing prefix $u'$

\item \hspace{0.25cm}
      $\widehat q_{u'}\gets\ActStop(u')$
      \hfill
      \textcolor{blue}{$\rhd$ Complete current prefix}

\item \hspace{0.25cm}
      \textbf{If}
      $i\geq\PruneFrac N_{\mathrm{rollout}}$,
      containment holds, and
      $\ResRS{\widehat q_{u'}}>
      >\max_y\Incumbent{y}$

\item \hspace{0.5cm}
      Prune $u'$
      \hfill
      \textcolor{blue}{$\rhd$ Sensitivity-guided pruning}

\item \hspace{0.25cm} \textbf{EndIf}

\item \textbf{EndFor}

\item Continue until $\ActStop()$ produces terminal query $q$
and trajectory $\tau$

\end{enumerate}

\vspace{1mm}
\textit{Evaluation}

\begin{enumerate}[nosep, label=\arabic*.]
\setcounter{enumi}{10}

\item \textbf{If} $q$ is invalid, nonexecutable, or not DP-compatible

\item \hspace{0.25cm} Backpropagate $0$ and \textbf{Return}

\item \textbf{EndIf}

\item $y\gets\operatorname{bucket}(q(\SynDB))$;
      $C_y\gets C_y\cup\{q\}$
\hfill
\textcolor{blue}{$\rhd$ Execution cluster}

 \item Compute
$\displaystyle
\RewCorr{q}\gets
\frac{|C_{y(q)}|-1}{M-1}$
\hfill
\textcolor{blue}{$\rhd$ Reward execution agreement}

\item Compute $\ResRS{q}$ and cache
$(k^\star,\mathbf{s}^\star,\{T_E\})$

\item Compute $W(U)$ from cached
$(k^\star,\mathbf{s}^\star,\{T_E\})$
\hfill
\textcolor{blue}{$\rhd$ Relation-subset sensitivity}

\item \textbf{For} each relation $r\in V(q)$

\item \hspace{0.25cm}
$\displaystyle
\phi_r(q)\gets
\sum_{U\subseteq V\setminus\{r\}}
\frac{|U|!(|V|-|U|-1)!}{|V|!}
\big[W(U\cup\{r\})-W(U)\big]$
\hfill
\textcolor{blue}{$\rhd$ Average marginal contribution}

\item \textbf{EndFor}

\end{enumerate}

\vspace{1mm}
\textit{Backpropagation}

\begin{enumerate}[nosep, label=\arabic*.]
\setcounter{enumi}{20}

\item \textbf{For} each $(v_{t-1},a_t)\in\tau$

\item \hspace{0.25cm}
$G_t\gets
-\lambda_s\nu(\phi_{r_t})$
if $a_t=\ActAddRelation(r_t,J_t,F_t)$;
otherwise $G_t\gets0$
\hfill
\textcolor{blue}{$\rhd$ Sensitivity credit}

\item \hspace{0.25cm}
$Q(v_{t-1},a_t)
\gets
Q(v_{t-1},a_t)+\RewCorr{q}+G_t$

\item \hspace{0.25cm}
$N(v_{t-1},a_t)\gets N(v_{t-1},a_t)+1$

\item \hspace{0.25cm}
$N(v_{t-1})\gets N(v_{t-1})+1$

\item \textbf{EndFor}

\item Update $\Incumbent{y}$
\hfill
\textcolor{blue}{$\rhd$ Best formulation so far}

\end{enumerate}

\end{mybox2}
\subsection{Benchmarks}\label{app:benchmark}
\noindent\textbf{Differential Privacy-Aware Text2SQL Benchmark (\dpdataset{}).} Existing Text2SQL benchmarks do not identify the lowest-sensitivity formulation among correct alternatives, despite sensitivity variation being common: 54.5\% of DP-supported BIRD questions and 66.6\% of Ambrosia questions produce candidates with different sensitivities (App.~\ref{}). We construct \dpdataset{} by identifying alternative join paths from the foreign-key graphs of BIRD, Spider, Spider2-Lite, AmbiQT and EHRSQL, and use Claude 5 Sonnet to generate alternate SQL formulations and retain correct samples. We compute their RS, and define the minimum-sensitivity formulation as $q^\star$. The benchmark contains 3,673 training questions from BIRD (3,063), AmbiQT (486), and EHRSQL (124), and 886 disjoint test questions from BIRD (415), Spider (362), AmbiQT (64), Spider2-Lite (25), and EHRSQL (20). Alternative formulations differ in Residual Sensitivity by a median of $10.5\times$, demonstrating substantial opportunity for reducing DP noise through SQL formulation.
\begin{table}[]
\centering
\caption{\textbf{Ablation of \sysname{}'s search components} on XiYanSQL-32B.
Each row disables one component of the Full configuration in isolation, with AlphaSQL included as a reference.
$\Delta$ is the change in Joint relative to Full.}
\label{tab:ablation}
\setlength{\tabcolsep}{3.5pt}
\begin{tabular}{lrrrrr}
\toprule
Method & Acc. & Is-Min & Joint & $\Delta$ & Lat. \\
\midrule
Full (\sysname{})      & 81.8 & 76.6 & 75.3 & 0.0     & 159.2 \\
w/o warm-start         & 79.7 & 68.3 & 65.2 & $-10.1$ & 133.0 \\
w/o pruning            & 81.8 & 76.3 & 74.5 & $-0.8$  & 235.4 \\
w/o sensitivity prior  & 81.2 & 74.7 & 73.8 & $-1.5$  & 152.9 \\
w/o sensitivity reward & 81.0 & 66.1 & 64.7 & $-10.6$ & 178.7 \\
direct RS reward       & 81.0 & 66.4 & 64.6 & $-10.7$ & 172.1 \\
\bottomrule
\end{tabular}
\end{table}
\noindent\textbf{Contextual Integrity-Aware Text2SQL Benchmark
(\dataset{1}, \dataset{2}, \dataset{3}).}
No existing Text-to-SQL benchmark captures contextual-integrity policies. We therefore construct one by pairing three real policies with domain-matched Text-to-SQL datasets: HIPAA with 56 EHRSQL queries (healthcare privacy), CCPA with 75 BIRD Car Retail queries (consumer privacy), and COPPA with 21 BIRD Computer Student Records queries (children's online privacy), for 152 queries total. We construct the benchmark in three steps.

\noindent \textit{Policy extraction.}
We first use an LLM to extract CI norms specifying roles, information, purposes, conditions, and permit/prohibit decisions. Since prior work extensively studies this task~\cite{mireshghallah2024can,ghalebikesabi2025privacy,tsai2025contextual,yi2025privacy}, we treat extraction as preprocessing.

\noindent \textit{Policy grounding.}
We manually annotate ground-truth mappings from policy concepts to their database representations: roles to identity attributes, information types to data attributes, purposes to database values, and conditions to SQL predicates. Concepts absent from the database are marked unmapped. Table~\ref{tab:policy-schema-annotation} summarizes these annotations.

\begin{table}[t]
\centering
\small
\begin{tabular}{lrrrr}
\toprule
Policy & Roles & Data & Purposes & Conditions \\
\midrule
CCPA  & 21 & 31 & 31 & 5  \\
COPPA & 12 & 48 & 28 & 15 \\
HIPAA & 44 & 7  & 6  & 11 \\
\bottomrule
\end{tabular}
\caption{Ground-truth policy-to-schema mappings.}
\label{tab:policy-schema-annotation}
\end{table}

\noindent\textit{Query grounding.}
We sample concrete sender and querier identities and purposes from the database and policy to form runtime contexts. For each query-context pair, we use the mappings to identify applicable norms, evaluate their conditions for the queried subjects, and apply their permit/prohibit decisions to obtain the ground-truth authorized output. If no norm applies, release is denied.

This process produces 495,072 query-context pairs, from which we sample 150 per policy for a 450-instance benchmark. Following prior work~\cite{squid}, we also generate textual narrations of each database. The resulting healthcare, consumer-privacy, and children's-privacy datasets are denoted \dataset{1}, \dataset{2}, and \dataset{3}, respectively. We use \dpdataset{} for RQ1 and \dataset{1}--\dataset{3} for RQ2--RQ4.

\subsection{Evaluation (Cntd.)}\label{app:eval}

\subsubsection{RQ1 (Cntd.)}
\noindent\textbf{Noise reduction.}
Fig.~\ref{fig:noise-reduction} (App.~\ref{app:eval}) reports the geometric mean reduction in DP noise achieved by \sysname{}. Across models, \sysname{} achieves $4.1$--$9.6\times$ lower noise than zero-shot and sensitivity prompting, and $13.3$--$33.5\times$ lower noise than AlphaSQL. 
\\\noindent\textbf{Ablation.}
We ablate \sysname{}'s four search components on XiYanSQL-32B: warm start, pruning, sensitivity prior, and Shapley-credit reward. We also replace Shapley credit with a coarse residual-sensitivity bonus $\frac{1}{1+\log(\mathrm{RS})}$. Table~\ref{tab:ablation} shows that warm start and Shapley credit have the largest impact, reducing Joint by 10.1 and 10.6 points when removed. Pruning reduces latency by 32.4\%. Replacing Shapley credit with the coarse RS bonus reduces Joint by 10.7 points, performing even worse than removing the sensitivity reward entirely.

\noindent\textbf{Cost.}
Table~\ref{tab:inference_cost} reports the token cost of \sysname{} vs AlphaSQL. \sysname{} is more efficient than AlphaSQL, using $2.6$--$5.1\times$ fewer tokens across models.
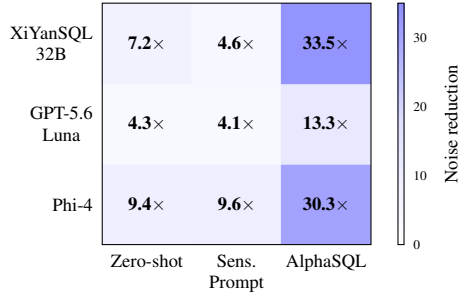
\begin{figure}[]
\centering
\begin{tikzpicture}
\begin{axis}[
    width=0.42\columnwidth,
    height=3.2cm,
    scale only axis,
    enlargelimits=false,
    axis on top,
    xtick={0,1,2},
    xticklabels={
        Zero-shot,
        {Sens.\\Prompt},
        AlphaSQL
    },
    ytick={0,1,2},
    yticklabels={
        {Phi-4},
        {GPT-5.6\\Luna},
        {XiYanSQL\\32B}
    },
    x tick label style={
        font=\scriptsize,
        align=center
    },
    y tick label style={
        font=\scriptsize,
        align=center
    },
    tick style={draw=none},
    point meta min=0,
    point meta max=35,
    colormap={softblue}{
        color(0cm)=(white);
        color(0.5cm)=(blue!12);
        color(1cm)=(blue!45)
    },
    colorbar,
    colorbar style={
        ylabel={Noise reduction},
        ylabel style={font=\scriptsize},
        yticklabel style={font=\tiny},
        width=0.08cm,
    },
]

\addplot[
    matrix plot*,
    mesh/cols=3,
    point meta=explicit,
] coordinates {
    (0,0) [9.42]
    (1,0) [9.57]
    (2,0) [30.34]

    (0,1) [4.26]
    (1,1) [4.13]
    (2,1) [13.33]

    (0,2) [7.18]
    (1,2) [4.63]
    (2,2) [33.47]
};

\node[font=\scriptsize\bfseries] at (axis cs:0,0) {9.4$\times$};
\node[font=\scriptsize\bfseries] at (axis cs:1,0) {9.6$\times$};
\node[font=\scriptsize\bfseries] at (axis cs:2,0) {30.3$\times$};

\node[font=\scriptsize\bfseries] at (axis cs:0,1) {4.3$\times$};
\node[font=\scriptsize\bfseries] at (axis cs:1,1) {4.1$\times$};
\node[font=\scriptsize\bfseries] at (axis cs:2,1) {13.3$\times$};

\node[font=\scriptsize\bfseries] at (axis cs:0,2) {7.2$\times$};
\node[font=\scriptsize\bfseries] at (axis cs:1,2) {4.6$\times$};
\node[font=\scriptsize\bfseries] at (axis cs:2,2) {33.5$\times$};

\end{axis}
\end{tikzpicture}
\caption{Geometric mean reduction in DP noise relative to each baseline.}
\label{fig:noise-reduction}
\end{figure}
\begin{table}[]
\centering
\caption{\textbf{Inference cost of search-based Text2SQL methods.}
We report the average number of model calls and total tokens consumed per
question.}
\label{tab:inference_cost}
\begin{tabular}{llrr}
\toprule
Model & Method & Calls / Q. & Tokens / Q. \\
\midrule
GPT-5.6-Luna
    & AlphaSQL & 24.9 & 40.1K \\
    & \sysname{} & \textbf{6.6} & \textbf{11.2K} \\
\midrule
XiYanSQL-32B
    & AlphaSQL & 24.9 & 19.6K \\
    & \sysname{} & \textbf{6.6} & \textbf{3.8K} \\
\midrule
Phi-4
    & AlphaSQL & 24.9 & 11.4K \\
    & \sysname{} & \textbf{6.6} & \textbf{4.4K} \\
\bottomrule
\end{tabular}
\end{table}

\end{document}